\documentclass[notitlepage,a4,11pt]{article}

\usepackage{graphicx}

\usepackage{amsmath}%
\usepackage{amsfonts}%
\usepackage{amssymb}%
\usepackage{mathrsfs}
\usepackage{amsthm}

\usepackage{authblk}

\newcommand{\mc}{\mathcal}
\newcommand{\cp}{\times}

\newcommand{\di}{\nabla\cdot}
\newcommand{\cu}{\nabla\times} 
\newcommand{\JI}[1]{\bol{#1}\cdot\cu{\bol{#1}}}
\newcommand{\BN}[1]{\bol{#1}\cp\left(\cu{\bol{#1}}\right)}

\newcommand{\bol}{\boldsymbol}
\newcommand{\hb}[1]{\hat{\bol{#1}}}
\newcommand{\abs}[1]{\left\lvert{#1}\right\rvert}
\newcommand{\w}{\wedge}
\newcommand{\lr}[1]{\left({#1}\right)}

\newcommand{\mf}{\mathfrak}
\newcommand{\p}{\partial}

\newcommand{\ti}[1]{\textit{#1}}
\newcommand{\tb}[1]{\textbf{#1}}
\newcommand{\ov}[1]{\mkern 1.5mu\overline{\mkern-1.5mu#1\mkern-1.5mu}\mkern 1.5mu}

\newtheorem{theorem}{\textit{Theorem}}
\newtheorem{proposition}{\textit{Proposition}}

\begin{document}
\title{Local Representation and Construction\protect\\ of 
Beltrami Fields II.\protect\\ Solenoidal Beltrami fields and ideal MHD equilibria}
\author[1]{N. Sato} \author[1]{M. Yamada}
\affil[1]{Research Institute for Mathematical Sciences, \protect\\ Kyoto University, Kyoto 606-8502, Japan \protect\\ Email: sato@kurims.kyoto-u.ac.jp}
\date{\today}
\setcounter{Maxaffil}{0}
\renewcommand\Affilfont{\itshape\small}


\maketitle

\begin{abstract}
Object of the present paper is the local theory of solution for 
steady ideal Euler flows and ideal MHD equilibria. 
The present analysis relies on the Lie-Darboux theorem of differential geometry  
and the local theory of representation and construction of Beltrami fields
developed in \cite{Sato}. 
A theorem for the construction of harmonic orthogonal coordinates is proved.
Using such coordinates families of solenoidal Beltrami fields with different topologies are obtained
in analytic form. 
Existence of global solenoidal Beltrami fields satisfying prescribed boundary conditions while
preserving the local representation is considered.
It is shown that only singular solutions are admissible, 
an explicit example is given in a spherical domain, and a theorem on existence of singular solutions is proven.
Local conditions for existence of solutions, and
local representation theorems are derived for generalized Beltrami fields, 
ideal MHD equilibria, and general steady ideal Euler flows. 
The theory is applied to construct analytic examples.
\end{abstract}





\section{Introduction}
Let $\bol{w}\in C^{\infty}({{\Omega}})$ be a smooth vector field in a smoothly bounded domain $\Omega\subset\mathbb{R}^3$ with boundary $\p\Omega$. 
Object of the present paper is the following system of first order partial differential equations:
\begin{equation}
\bol{w}\cp\lr{\nabla\cp\bol{w}}=\nabla\lr{P+\kappa\, \bol{w}^2}~~~~{\rm in}~~\Omega,\label{Beq}
\end{equation}
where $P\in C^{\infty}\lr{\Omega}$ and $\kappa\in\mathbb{R}$ are given.
System \eqref{Beq} will be examined both independently, and with the additional divergence free condition
\begin{equation}
\di\bol{w}=0~~~~{\rm in}~~\Omega.\label{Sol}
\end{equation}
The boundary condition 
\begin{equation}
\bol{w}\cdot\bol{n}=0~~~~{\rm on}~~\p\Omega,\label{BC}
\end{equation}
where $\bol{n}$ is the unit outward normal to $\p\Omega$, will also be considered.

In the context of fluid dynamics \cite{Moffatt14}, equation \eqref{Beq}, together with the continuity equation $\di\lr{\rho\bol{w}}=0$, represents a steady Euler flow: $\bol{w}$ is the fluid velocity, $\kappa=1/2$, and $\nabla P=\rho^{-1}\,\nabla p+\nabla\phi$, where $\rho$ is the fluid density, $p=p\lr{\rho}$ the barotropic pressure, and $\phi$ an external potential. For a plasma system, $\bol{w}$ is the solenoidal magnetic field, $\kappa=0$, and $-P$ is the plasma pressure. Then, equations \eqref{Beq} and \eqref{Sol} describe ideal magnetohydrodynamic equilibria \cite{Moffatt85}. 

An important class of solutions $\bol{w}$ to \eqref{Beq} can be obtained when the right-hand side vanishes,
i.e. $\nabla P=-\kappa\,\nabla\bol{w}^2$. In such case, a vector field $\bol{w}\neq\bol{0}$ is called
a Beltrami field and satisfies the Beltrami condition
\begin{equation}
\cu\bol{w}=\hat{h}\,\bol{w}~~~~{\rm in}~~\Omega.\label{Beq0}
\end{equation}
The function $\hat{h}\in C^{\infty}\lr{\Omega}$ is known as the proportionality factor (or coefficient).
Since a Beltrami field is an eigenvector of the curl operator \cite{Yoshida90},
the proportionality factor $\hat{h}$ can be identified with the associated eigenvalue. 

While the occurrence of Beltrami flows in steady fluids is a non-trivial physical problem,
Beltrami fields are essential in the description of relaxed plasma states \cite{Taylor74,Taylor86}.
Such discrepancy originates from the different degree of conservation between fluid helicity and magnetic helicity. Indeed, due to the higher order derivatives appearing in the fluid helicity, the fluid kinetic energy
is expected to dissipate at a slower rate, 
and fluid Beltrami equilibria are not properly described, in general, by minimization of energy under the constraint of helicity \cite{Yoshida02} (see  
\cite{Scheeler14} for measurements of the extent of helicity conservation in 
non-ideal fluid systems). 
The situation is opposite in the magnetic case: Beltrami fields with constant proportionality factor are consistent
with a relaxation process where magnetic energy is minimized 
by keeping the magnetic helicity constant \cite{Woltjer58}.

From a mathematical standpoint, Beltrami fields 
serve as a tool to investigate chaotic properties \cite{Dombre86} 
and topological behavior of fluid flows \cite{Enciso15}, 
as well as the statistical mechanics of magnetohydrodynamic systems \cite{Ito96}.
When, instead of fluid models, the motion of a single particle is considered,
Beltrami fields arise in the form of antisymmetric operators that
act on the particle Hamiltonian to induce dynamics. 
Such `Beltrami operators' exhibit special statistical properties,
namely that they enable a formulation of statistical mechanics
in the absence of canonical phase space (see \cite{Sato18}).

Although Beltrami fields find several applications due to their
occurrence in fluid and magnetofluid systems, 
the solvability of the Beltrami equation \eqref{Beq0} 
for given $\hat{h}\in C^{\infty}\lr{\Omega}$ poses a mathematical challenge.
First, a major distinction exists between the case of constant
proportionality factor, i.e. the case with $\nabla\hat{h}=\bol{0}$ in $\Omega$,
and the general situation in which $\hat{h}$ is a function in $\Omega$.
Observe that, if $\hat{h}$ is constant, then $\di\bol{w}=0$ when the Beltrami condition \eqref{Beq0} is satisfied.
In \cite{Yoshida90} it is proved that the Beltrami equation \eqref{Beq0} 
endowed with boundary condition \eqref{BC} and 
with constant proportionality factor $\hat{h}\in \mathbb{C}$ (note that $\hat{h}$ can be complex valued) always admits a nontrivial solution if $\Omega$ is a multiply connected smoothly bounded domain.
Here, a nontrivial solution is to be found in an appropriate subset
of the standard Sobolev space of order 1, $H^{1}\lr{\Omega}$. 
When $\hat{h}$ is allowed to be a function in the class $C^{6,\alpha}\lr{\Omega}$, it has been shown that the set of inhomogeneous
proportionality factors by which the Beltrami equation \eqref{Beq0} 
admits a solenoidal solution $\bol{w}\neq\bol{0}$, $\di\bol{w}=0$, has `measure' zero.
More precisely, $\bol{w}=\bol{0}$ for all $\hat{h}$ in a dense and open subset of $C^{7}\lr{\Omega}$ (see \cite{Enciso16}). If the solenoidal condition \eqref{Sol} is abandoned, weak solutions of system \eqref{Beq0} can be derived \cite{Kress77} (see also \cite{Kaiser00} and the discussion therein). However, when available, such solutions remain of pure mathematical interest.

If the solenoidal Beltrami field arising from equations \eqref{Beq0} and \eqref{Sol} is assumed to exhibit Euclidean symmetries, it is known that
equations \eqref{Beq0} and \eqref{Sol} reduce to a single nonlinear elliptic PDE
for the flux function, the Grad-Shafranov equation \cite{Grad58}.
The Grad-Shafranov equation can be derived also for system \eqref{Beq} together with \eqref{Sol} if $\kappa=0$, i.e. for ideal magnetohydrodynamic equilibria with non-vanishing pressure gradients and Euclidean symmetries \cite{Edenstrasser80_1,Edenstrasser80_2}. Thus, in the presence of suitable symmetries,
the Grad-Shafranov equation provides a viable way to obtain
solenoidal Beltrami fields in analytic or numerical form. 
We refer the reader to \cite{Tassi08} for analytic examples of symmetric solutions, and to \cite{Hudson07,Hudson17} for 
numerical examples of force free equilibria in the context of 
magnetic plasma confinement.

A systematic method to construct Beltrami fields in analytic form has been
obtained in \cite{Sato}. This result relies on a local solution of the
Beltrami equation \eqref{Beq0}: the solution is found in a small neighborhood $U\subset\Omega$ and boundary conditions are discarded.  
The analysis relies on the helicity density
\begin{equation}
h=\JI{w},\label{hel}
\end{equation}
of the vector field $\bol{w}$. $h$ characterizes the topological properties of $\bol{w}$. A vanishing helicity density $h=0$ 
implies that a smooth vector field $\bol{w}$ is integrable in the sense of the Frobenius' theorem \cite{Frankel}, and thus admits a local representation
\begin{equation}
\bol{w}=\lambda\,\nabla C~~~~{\rm in}~~U,\label{Fth}
\end{equation}
with $\lambda,C\in C^{\infty}\lr{U}$.
Since any nontrivial Beltrami field satisfies $\cu\bol{w}=\hat{h}\,\bol{w}$ with 
$\cu\bol{w}\neq\bol{0}$, it is clear that $h\neq 0$.
In \cite{Sato} it is shown that under such circumstances 
a Beltrami field has the local representation
\begin{equation}
\bol{w}=\cos\theta\,\nabla\psi+\sin\theta\,\nabla\ell~~~~{\rm in}~~\Omega,\label{BTh}
\end{equation}
with $\ell,\psi,\theta\in C^{\infty}\lr{U}$ a smooth coordinate system satisfying the following geometric conditions,
\begin{subequations}\label{GC}
\begin{align}
&\cos{\theta}~\sin{\theta}\lr{\abs{\nabla\psi}^2-\abs{\nabla\ell}^2}=\lr{\nabla\ell\cdot\nabla\psi}\lr{\cos^2{\theta}-\sin^2{\theta}},\\
&\sin{\theta}~\nabla\ell\cdot\nabla\theta+\cos{\theta}~\nabla\psi\cdot\nabla\theta=0,
\end{align}
\end{subequations}
The proof of \eqref{BTh} and \eqref{GC} relies on the Darboux theorem,
which gives a local expression for closed differential 2-forms \cite{DeLeon89,Arnold89,Salmon17}. 

The local representation \eqref{BTh} has a number of consequences.
First, Beltrami flows admit two local invariants, the plane of the flow $\theta$ and the angular momentum-like quantity $L_{\theta}=\ell \cos\theta-\psi\sin\theta$. Furthermore, if $\di\bol{w}=0$, it can be shown that the proportionality factor satisfies $\hat{h}=\hat{h}\lr{\theta,L_{\theta}}$. Hence, $\hat{h}$ becomes a constant of the motion itself.
Since $\hat{h}$ represents the helicity density of the
normalized vector field $\hb{w}/\abs{\bol{w}}$, i.e. $\hat{h}=\hb{w}\cdot\cu\hb{w}$,  
the conservation of $\hat{h}$ physically correspond to conservation of (normalized) helicity. This shows that, at least at a local level, inhomogeneous proportionality factors do not increase the number of conservation laws,  
thus explaining the so called helical flow paradox \cite{Morgulis95}.
Finally, the local representation \eqref{BTh} naturally leads to the
following method for the construction of Beltrami fields: 
from conditions \eqref{GC} one sees that a Beltrami field
in the form of \eqref{BTh} can be obtained by finding an orthogonal coordinate system $\lr{\ell,\psi,\theta}$ such that two scale factors are equal, $\abs{\nabla\ell}=\abs{\nabla\psi}$.    
If the coordinates $\ell$ and $\psi$ are harmonic,
the resulting Beltrami field is solenoidal.
We refer the reader to \cite{Sato} for a list of explicit examples.  

Aim of the present paper is to extend the results of \cite{Sato}
in the following ways: first, in section 2 we prove a theorem that 
provides harmonic orthogonal coordinate systems (a precise definition of such coordinates will be given later on).
These coordinates can be used to construct solenoidal Beltrami fields
according to the method of \cite{Sato}.
Examples of families of solenoidal Beltrami fields with different topologies are given in section 3. Here, the global problem is also considered.
More precisely, we examine the conditions under which 
the local representation \eqref{BTh} can be prolonged up 
to the boundary, so that boundary conditions \eqref{BC} are satisfied 
(in the following, we shall refer to a global solution in the sense
that it satisfies prescribed boundary conditions while maintaining its local representation).  
It is shown that the availability of a global Beltrami field 
is limited to singular solutions, i.e. solutions which are regular almost
everywhere in the domain of interest. An example of singular global solution
within a spherical domain is given.
Finally, the general system, equation \eqref{Beq}
is examined. Local conditions for existence of solutions
are derived. In section 4, the case of generalized Beltrami fields,
i.e. solutions of \eqref{Beq} for $\nabla P=\bol{0}$, is discussed.
In section 5, we consider ideal MHD equilibria, while
steady Euler flows are discussed in section 6.
In section 7 we draw our conclusions.

\section{Construction of solenoidal Beltrami fields}
Purpose of the present section is to derive a method to
obtain a set of orthogonal coordinates $\lr{\ell,\psi,\theta}\in C^{\infty}\lr{\Omega}$ with the property that
\begin{equation}
\abs{\nabla\ell}=\abs{\nabla\psi},~~~~\Delta\ell=0,~~~~\Delta\psi=0~~~~{\rm in}~~\Omega,\label{hC}
\end{equation} 
where $\Omega\subset\mathbb{R}^3$ is some appropriate domain.
We shall refer to an orthogonal coordinate system $\lr{\ell,\psi,\theta}$ satisfying \eqref{hC} as harmonic.
If such a coordinate system is found, then, from \cite{Sato}, we know that $\bol{w}=\cos\theta\,\nabla\psi+\sin\theta\,\nabla\ell$ is a solenoidal Beltrami field in $\Omega$.

Before giving the formal statement, it is useful to explain when and how
harmonic orthogonal coordinates can constructed. 
First, suppose that $\lr{\alpha,\beta,\gamma}\in C^{\infty}\lr{\Omega}$ 
is an orthogonal coordinate system in a domain $\Omega\subset\mathbb{R}^3$.  
Let $V$ denote the two-dimensional set $V=\left\{\lr{\alpha\lr{\bol{x}},\beta\lr{\bol{x}}}~|~\bol{x}\in U\right\}$. 
In $V$, consider the
two-dimensional Laplace equation
\begin{equation}
f_{\alpha\alpha}+f_{\beta\beta}=0~~~~{\rm in}~~V.\label{2Lap}
\end{equation}
Here, a lower index indicates a partial derivative.
If the domain $V$ has a sufficiently regular boundary, 
equation \eqref{2Lap} can be solved \cite{Gilbarg01,Evans10}.
The solution is a harmonic function $f\lr{\alpha,\beta}$.
Furthermore, if the domain $V$ is simply connected, from standard results of complex analysis, the function $f$
also admits a harmonic conjugate $g$ bearing the Cauchy-Riemann equations
\begin{equation}
f_{\alpha}=g_{\beta},~~~~f_{\beta}=-g_{\alpha}~~~~{\rm in}~~V.\label{CaRi}
\end{equation}
These relationships also imply that
\begin{equation}
f_{\alpha}^2+f_{\beta}^2=g_{\alpha}^2+g_{\beta}^2~~~~{\rm in}~~V.
\end{equation}
Hence, the functions $f$ and $g$ are harmonic and their gradients have equal length in the metric induced by $\alpha$ and $\beta$.   
Note that both $f$ and $g$ are well defined in the whole $\Omega$ since
they do not explicitly depend on the variable $\gamma$. 
Therefore, our original problem \eqref{hC} is solved if the properties 
of $f$ and $g$ remain valid in the whole $\Omega$, namely that
they are harmonic and the length of their gradients is equal also with respect to the Euclidean metric of $\mathbb{R}^3$.
Our task is to determine the geometric conditions on the coordinates
$\lr{\alpha,\beta,\gamma}$ that allow such transfer of properties. 
On this regard, we have the following:

\begin{theorem}\label{thm1}
Let $\lr{\zeta,\beta,\gamma}\in C^{\infty}\lr{\Omega}$ be a smooth orthogonal coordinate system in a smoothly bounded simply connected domain $\Omega\subset\mathbb{R}^3$. Assume that
\begin{equation}
\p_{\beta}\lr{\frac{\abs{\nabla\zeta}}{\abs{\nabla\beta}}}=\p_{\gamma}\lr{\frac{\abs{\nabla\zeta}}{\abs{\nabla\beta}}}=\p_{\zeta}\abs{\nabla\gamma}=\p_{\beta}\abs{\nabla\gamma}=0.\label{hyp1}
\end{equation}
Then, there exist smooth functions $\lr{\ell,\psi,\theta}\in C^{\infty}(\Omega)$ 
such that the vector field
\begin{equation}
\bol{w}=\cos\left[{\sigma\lr{\theta}}\right]\,\nabla\psi+\sin\left[{\sigma\lr{\theta}}\right]\,\nabla\ell,\label{solbf}
\end{equation}
is a smooth solenoidal Beltrami field in $\Omega$.
Here, $\sigma=\sigma\lr{\theta}$ is any smooth function of the variable $\theta$.
\end{theorem}
\begin{proof}
The Laplacian $\Delta f$ of a function $f\in C^{\infty}\lr{\Omega}$ with respect to the standard Euclidean metric of $\mathbb{R}^3$ is defined by the identity
\begin{equation}
\mf{L}_{\nabla f}dx\w dy\w dz=\Delta f\,dx\w dy\w dz,
\end{equation}
where $\mf{L}$ is the Lie derivative. On the other hand, observe that
\begin{equation}
\begin{split}
\mf{L}_{\nabla f}&dx\w dy\w dz=di_{\nabla f}\,J^{-1}d\zeta\w d\beta\w d\gamma\\
=&d\left[J^{-1}\lr{\frac{f_{\zeta}}{\abs{\p_\zeta}^2}d\beta\w d\gamma-\frac{f_{\beta}}{\abs{\p_\beta}^2}d\zeta\w d\gamma+\frac{f_{\gamma}}{\abs{\p_\gamma}^2}d\zeta\w d\beta}\right]\\=&J\left[\frac{\p}{\p\zeta}\lr{\frac{\abs{\nabla\zeta}^2}{J}f_{\zeta}}
+\frac{\p}{\p\beta}\lr{\frac{\abs{\nabla\beta}^2}{J}f_{\beta}}
+\frac{\p}{\p\gamma}\lr{\frac{\abs{\nabla\gamma}^2}{J}f_{\gamma}}
\right]dx\w dy \w dz.
\end{split}
\end{equation}
Here, $J=\nabla\zeta\cdot\nabla\beta\cp\nabla\gamma$ is the Jacobian of
the coordinate change $\lr{\zeta,\beta,\gamma}\mapsto\lr{x,y,z}$.
Therefore,
\begin{equation}
\Delta f=J\left[\frac{\p}{\p\zeta}\lr{\frac{\abs{\nabla\zeta}^2}{J}f_{\zeta}}
+\frac{\p}{\p\beta}\lr{\frac{\abs{\nabla\beta}^2}{J}f_{\beta}}
+\frac{\p}{\p\gamma}\lr{\frac{\abs{\nabla\gamma}^2}{J}f_{\gamma}}
\right].
\end{equation}
If $f_{\gamma}=0$, Laplace's equation $\Delta f=0$ thus reduces to
\begin{equation}
\begin{split}
&\frac{\p}{\p\zeta}\lr{\frac{\abs{\nabla\zeta}^2}{J}f_{\zeta}}
+\frac{\p}{\p\beta}\lr{\frac{\abs{\nabla\beta}^2}{J}f_{\beta}}\\&=
\frac{\p}{\p\zeta}\lr{\frac{\abs{\nabla\zeta}}{\abs{\nabla\beta}\abs{\nabla\gamma}}f_{\zeta}}
+\frac{\p}{\p\beta}\lr{\frac{\abs{\nabla\beta}}{\abs{\nabla\zeta}\abs{\nabla\gamma}}f_{\beta}}=0~~~~{\rm in}~~\Omega.\label{D2}
\end{split}
\end{equation}
Here, the orientation of the coordinate system $\lr{\zeta,\beta,\gamma}$ is assumed to be such that $J/\abs{J}=1$.
By hypothesis, the metric coefficient $\abs{\nabla\gamma}\neq 0$ is independent of $\zeta$ and $\beta$, while $\abs{\nabla\zeta}/\abs{\nabla\beta}\neq 0$ is a function of $\zeta$. This implies that equation \eqref{D2} can be written as
\begin{equation}
\frac{\abs{\nabla\zeta}}{\abs{\nabla\beta}}\frac{\p}{\p\zeta}\lr{\frac{\abs{\nabla\zeta}}{\abs{\nabla\beta}}f_{\zeta}}+f_{\beta\beta}=0~~~~{\rm in }~~\Omega,\label{D3}
\end{equation}
Next, we define the smooth function
\begin{equation}
\alpha=\int{\frac{\abs{\nabla\beta}}{\abs{\nabla\zeta}}}\,d\zeta.
\end{equation}
Note that $\alpha$ is a strictly
monotonic function of $\zeta$. Therefore, it is invertible and can be used as coordinate in place of the variable $\zeta$. 
In terms of $\alpha$ and $\beta$ equation \eqref{D3} becomes
\begin{equation}
f_{\alpha\alpha}+f_{\beta\beta}=0~~~~{\rm in}~~V.\label{redlap}
\end{equation}
where the set $V=\left\{\lr{\alpha\lr{\bol{x}},\beta\lr{\bol{x}}}~|~\bol{x}\in \Omega\right\}$ is the image of the smooth coordinates $\lr{\alpha,\beta}$. 
Evidently, $V$ is a smoothly bounded simply connected domain. 

Observe that a solution $f$ of equation \eqref{redlap} automatically satisfies $\Delta f=0$ in $\Omega$ as a consequence of its equivalence with equation \eqref{D3}.
Since $V$ is simply connected, the smooth solution $f$ admits an harmonic conjugate $g$ such that $f_{\alpha}=g_{\beta}$ and $f_{\beta}=-g_{\alpha}$ in $V$, and $\Delta g=0$ in $\Omega$. It follows that, 
\begin{equation}
\abs{\nabla f}^2=f^{2}_{\alpha}\abs{\nabla\alpha}^2+f^2_{\beta}\abs{\nabla\beta}^2=g_{\beta}^2\alpha_{\zeta}^2\abs{\nabla\zeta}^2+g_{\zeta}^2\zeta_{\alpha}^2\abs{\nabla\beta}^2=\abs{\nabla g}^2~~~~{\rm in }~~\Omega.
\end{equation}
Here, the fact that $\alpha_{\zeta}=\zeta_{\alpha}^{-1}=\abs{\nabla\beta}/\abs{\nabla\zeta}$ was used.
This shows that, by setting $\lr{\ell,\psi,\theta}=\lr{f,g,\gamma}$, the gradients of the coordinates $\lr{\ell,\psi,\theta}$ are orthogonal to each other, and  $\abs{\nabla\ell}=\abs{\nabla\psi}$. 
It follows that the vector field \eqref{solbf} is a solenoidal Beltrami field in $\Omega$. Recalling that $\hat{h}=h/w^2$, the proportionality factor is $\hat{h}=\sigma_{\theta}\abs{\nabla\theta}$.
\end{proof}

The result of theorem \ref{thm1} can be rephrased as follows:
given 
an orthogonal coordinate system $\lr{\zeta,\beta,\gamma}$ with the property \eqref{hyp1}, it is always possible to determine an associated solenoidal
Beltrami field $\bol{w}=\cos\theta\,\nabla\psi+\sin\theta\,\nabla\ell$.
One can verify that several coordinate systems satisfy the requirements
of theorem \ref{thm1}. We will see that for this reason 
theorem \ref{thm1} can be used to construct
solenoidal Beltrami fields with different topologies in analytic form.

\section{Families of solenoidal Beltrami fields}
In this section, we apply theorem \ref{thm1} to construct solenoidal
Beltrami fields. 

First, observe that Cartesian, cylindrical, and spherical coordinates satisfy
the hypothesis of theorem \ref{thm1}.
In Cartesian coordinates $\lr{x,y,z}$ all scale factors are equal,
$\abs{\nabla x}=\abs{\nabla y}=\abs{\nabla z}=1$. Hence, any permutation of the
coordinates can be identified with $\lr{\zeta,\beta,\gamma}$. In particular, $\alpha=\zeta$, and, by setting $\lr{\zeta,\beta,\gamma}=\lr{x,y,z}$, theorem \ref{thm1} guarantees that the vector field
\begin{equation}
\bol{w}=\cos\left[ \sigma\lr{z}\right]\,\nabla\psi\lr{x,y}+\sin\left[ \sigma\lr{z}\right]\,\nabla\ell\lr{x,y},
\end{equation}
is a solenoidal Beltrami field provided that $\ell\lr{x,y}$ and $\psi\lr{x,y}$ are harmonic conjugate functions in the $\lr{x,y}$ variables. Examples
are $\ell=x$ and $\psi=y$, $\ell=-e^x \sin y$ and $\psi=e^x\cos{y}$, $\ell=\sinh y\sin x$ and $\psi=\cosh y \cos x$, or $\ell=\frac{\sin y}{\cosh x-\cos y}$ and $\psi=\frac{\sinh x}{\cosh x-\cos y}$. 
 
In cylindrical coordinates $\lr{r,\phi,z}=\lr{\sqrt{x^2+y^2},\arctan\lr{\frac{y}{x}},z}$, define $\lr{\zeta,\beta,\gamma}=\lr{r,\phi,z}$. Then,
\begin{equation}
\p_{\phi}\lr{\frac{\abs{\nabla r}}{\abs{\nabla\phi}}}=\p_{z}\lr{\frac{\abs{\nabla r}}{\abs{\nabla\phi}}}=\p_{r}\abs{\nabla z}=\p_{\phi}\abs{\nabla z}=0.
\end{equation} 
Furthermore, $\alpha=\int\frac{dr}{r}=\log{r}+c$, with $c\in\mathbb{R}$ a constant of integration. We set $c=0$.
It follows that the vector field
\begin{equation}
\bol{w}=\cos\left[\sigma\lr{z}\right]\,\nabla\psi\lr{\log r,\phi}+\sin\left[\sigma\lr{z}\right]\,\nabla\ell\lr{\log r,\phi},
\end{equation}
is a solenoidal Beltrami field provided that $\ell\lr{\log r,\phi}$ and
$\psi\lr{\log r,\phi}$ are harmoninc conjugate functions in the $\lr{\log r,\phi}$ variables. Examples are $\ell=\log r$ and $\psi=\phi$, $\ell=-r\sin\phi$ and $\psi=r\cos\phi$, $\ell=\sinh\lr{\log r}\sin\phi$ and $\psi=\cosh\lr{\log r}\cos\phi$, or $\ell=\frac{\sin\phi}{\cosh\lr{\log r}-\cos\phi}$ and $\psi=\frac{\sinh\lr{\log r}}{\cosh\lr{\log r}-\cos\phi}$.

Let $\Omega=S\subset\mathbb{R}^3$ be a sphere of radius $R_{0}>0$ centered at the origin $\bol{x}=\bol{0}$. We denote by
\begin{equation}
\lr{R,\vartheta,\phi}=\lr{\sqrt{x^2+y^2+z^2},\arccos\lr{\frac{z}{\sqrt{x^2+y^2+z^2}}},\arctan\lr{\frac{y}{x}}},
\end{equation}
the standard spherical coordinate system.
The scale factors are
\begin{equation}
\abs{\nabla R}^2=1,~~~~\abs{\nabla\vartheta}^2=\frac{1}{R^2},~~~~
\abs{\nabla\phi}^2=\frac{1}{R^2\sin^2\vartheta}.
\end{equation}
Observe that, by setting $\lr{\zeta,\beta,\gamma}=\lr{\vartheta,\phi,R}$, one has
\begin{equation}
\p_{\phi}\lr{\frac{\abs{\nabla \vartheta}}{\abs{\nabla\phi}}}=\p_{R}\lr{\frac{\abs{\nabla \vartheta}}{\abs{\nabla\phi}}}=\p_{\vartheta}\abs{\nabla R}=\p_{\phi}\abs{\nabla R}=0.
\end{equation}
The variable $\alpha$ can be evaluated to be 
\begin{equation}
\alpha\lr{\vartheta}=\int{\frac{d\vartheta}{\sin\vartheta}}=\log\lr{\frac{\sin\vartheta}{1+\cos\vartheta}}+c,
\end{equation}
with $c\in\mathbb{R}$ a constant of integration. We set $c=0$.
It follows that the vector field
\begin{equation}
\bol{w}=\cos\left[\sigma\lr{R}\right]\,\nabla\psi\lr{\alpha\lr{\vartheta},\phi}+\sin\left[\sigma\lr{R}\right]\,\nabla\ell\lr{\alpha\lr{\vartheta},\phi},\label{SB}
\end{equation}
is a solenoidal Beltrami field provided that $\ell\lr{\alpha\lr{\vartheta},\phi}$ and $\psi\lr{\alpha\lr{\vartheta},\phi}$ are harmonic conjugate functions in the variables $\lr{\alpha\lr{\vartheta},\phi}$. Examples are
$\ell=\alpha\lr{\vartheta}$ and $\psi=\phi$, $\ell=-\frac{\sin\vartheta\sin\phi}{1+\cos\vartheta}$ and $\psi=\frac{\sin\vartheta\cos\phi}{1+\cos\vartheta}$, $\ell=\sinh\alpha\lr{\vartheta}\sin\phi$ and $\psi=\cosh\alpha\lr{\vartheta}\cos\phi$, or $\ell=\frac{\sin\phi}{\cosh\alpha\lr{\vartheta}-\cos\phi}$ and $\psi=\frac{\sinh\alpha\lr{\vartheta}}{\cosh\alpha\lr{\vartheta}-\cos\phi}$.

Figure \ref{fig1} shows a plot of the vector field \eqref{SB} for
$\sigma\lr{R}=R$, $\ell\lr{\alpha\lr{\vartheta},\phi}=\phi$, and $\psi\lr{\alpha\lr{\vartheta},\phi}=\alpha\lr{\vartheta}$ on the level sets of the invariants $R$ and $L_{R}=\phi\cos R-\alpha\lr{\vartheta}\sin R$.

\begin{figure}[h]
\hspace*{-0cm}\centering
\includegraphics[scale=0.3]{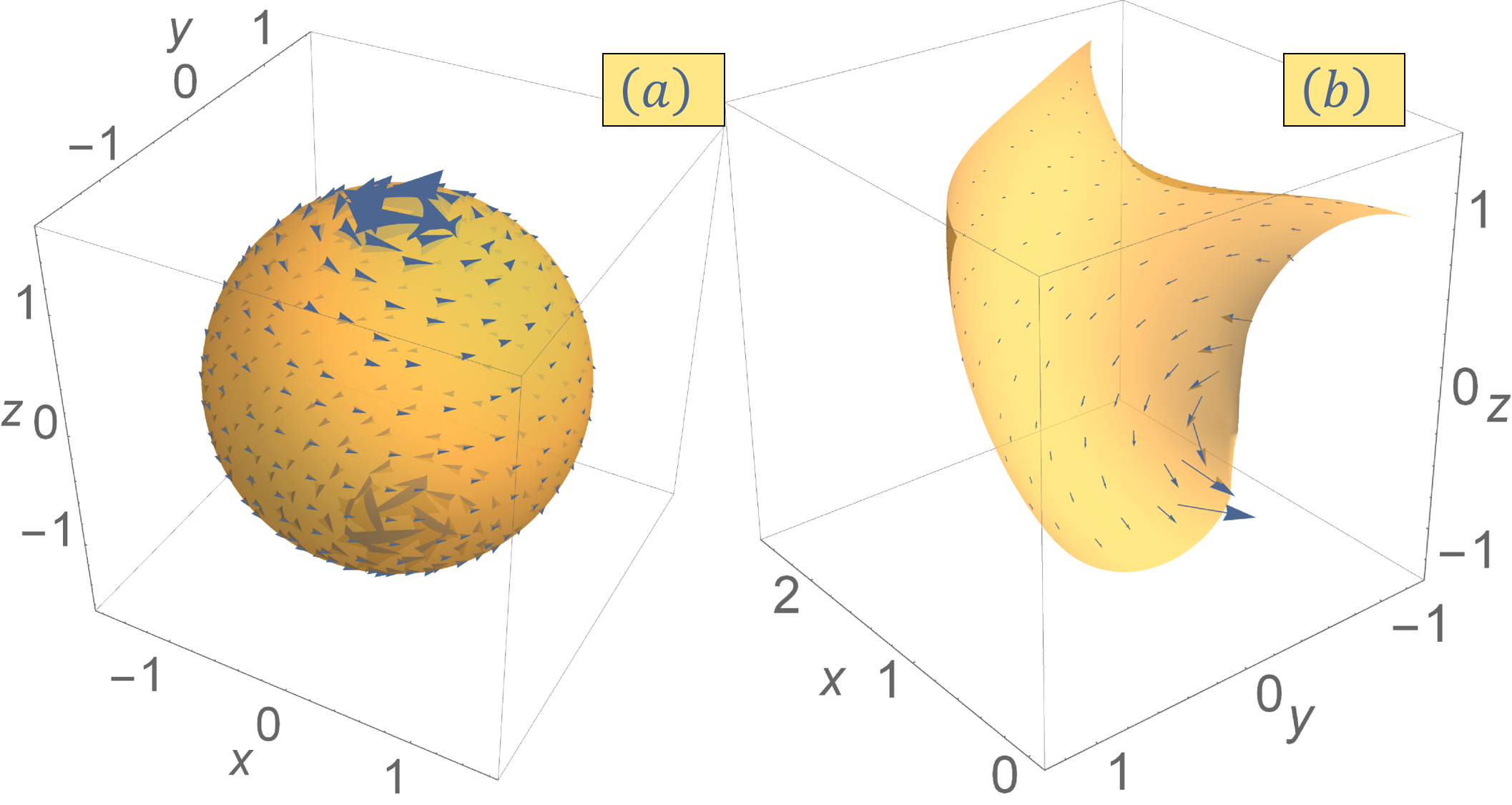}
\caption{\footnotesize (a): Plot of the solenoidal Beltrami field \eqref{SB} on the surface $R=1$. (b): Plot of the solenoidal Beltrami field \eqref{SB} on the surface $L_{R}=0.5$. Here, $\sigma\lr{R}=R$, $\ell\lr{\alpha\lr{\vartheta},\phi}=\phi$, and $\psi\lr{\alpha\lr{\vartheta},\phi}=\alpha\lr{\vartheta}$.}
\label{fig1}
\end{figure}

Since $\nabla\vartheta\cdot\nabla R=\nabla\phi\cdot\nabla R=0$, we have $\bol{w}\cdot\nabla R=0$. Hence, the vector field \eqref{SB} is orthogonal to
the unit outward normal $\bol{n}=\nabla R$ to the surface of the sphere $S$, i.e. it is a solution of the boundary value problem
\begin{equation}  
\begin{split}
&\BN{w}=\bol{0},~~~~\di\bol{w}=0~~~~{\rm in }~~S,\\
&\bol{w}\cdot\nabla R=0~~~~{\rm on}~~\p S.
\end{split}
\end{equation}
Here $\p S$ is the boundary of the sphere $S$.
This fact deserves some comments. 
At first glance, there may seem to be a contradiction because 
both $\vartheta$ and $\phi$ appear to be solutions of the boundary value problem
\begin{equation}
\Delta f=0~~~~{\rm in}~~S,~~~~\nabla f\cdot\nabla R=0~~~~{\rm on}~~\p S.\label{LS}
\end{equation}
Indeed, from standard results of functional analysis, regular
solutions of \eqref{LS} are unique up to constants.
A careful examination of the solutions $\phi$ and $\vartheta$ shows however
that no contradiction exists because the derived solutions are singular, specifically $\abs{\nabla\theta}^2=1/R^2$ diverges when $R\rightarrow 0$, and $\abs{\nabla\phi}^2=1/R^2\sin^2\theta$ diverges when $R\rightarrow 0$ or $\sin\vartheta\rightarrow 0$.

This result suggests a method to construct singular solutions to the boundary value problem
\begin{equation}  
\begin{split}
&\BN{w}=\bol{0},~~~~\di\bol{w}=0~~~~{\rm in }~~\Omega,\\
&\bol{w}\cdot\nabla \nu\lr{\theta,L_{\theta}}=0~~~~{\rm on}~~\p \Omega,\label{BVP}
\end{split}
\end{equation}
where $\Omega$ is a smoothly bounded simply connected domain with unit outward normal $\bol{n}=\nabla\nu\lr{\theta,L_{\theta}}/\abs{\nabla\nu\lr{\theta,L_{\theta}}}$. Here, $\nu\lr{\theta,L_{\theta}}\in C^{\infty}\lr{\Omega}$ is a smooth function of the invariants $\theta$ and $L_{\theta}$ of the Beltrami flow $\bol{w}=\cos\theta\,\nabla\psi+\sin\theta\,\nabla\ell$.
Indeed, while there is no hope of finding two distinct regular solutions $\ell$ and $\psi$ to the boundary value problem
\begin{equation}
\Delta f=0~~~~{\rm in}~~\Omega,~~~~\nabla f\cdot\nabla \nu\lr{\theta,L_{\theta}}=0~~~~{\rm on}~~\p\Omega,
\end{equation}
distinct singular solutions may exist, as in the spherical case. 
By singular solution, we mean a solution that is regular in the domain of interest, exception made for a set of measure zero. 
We have:

\begin{proposition}\label{prop1}
Let $\Omega\subset\mathbb{R}^3$ be a smoothly bounded
simply connected domain with boundary $\p\Omega$. Let $\bol{n}=\nabla\nu/\abs{\nabla\nu}$ be the unit outward normal to $\p\Omega$, with $\nu\in C^{\infty}\lr{\Omega'}$ and $\Omega\subset\subset\Omega'\subset\mathbb{R}^3$. Suppose that there exist 
orthogonal coordinates $\lr{\ell,\psi,\theta}$ smooth almost everywhere in $\Omega'$ such that
\begin{subequations}\label{SingBVP}
\begin{align}
&\Delta\ell=0~~~~{\rm a.e.~~in}~~\Omega,~~~~\nabla\ell\cdot\nabla\nu=0~~~~{\rm a.e.~~on}~~\p\Omega,\label{p1eq1}\\
&\Delta\psi=0~~~~{\rm a.e.~~in}~~\Omega,~~~~\nabla\psi\cdot\nabla\nu=0~~~~{\rm a.e.~~on}~~\p\Omega,\label{p1eq2}\\
&\abs{\nabla\ell}=\abs{\nabla\psi}~~~~{\rm a.e.~~in}~~\Omega,\\
&\nu=\nu\lr{\theta,L_{\theta}}~~~~{\rm in}~~\Omega',
\end{align}
\end{subequations}
with $L_{\theta}=\ell\cos\theta-\psi\sin\theta$. 
Then, the vector field
\begin{equation}
\bol{w}=\cos\left[\sigma\lr{\theta}\right]\,\nabla\psi+\sin\left[\sigma\lr{\theta}\right]\,\nabla\ell,\label{P1BF}
\end{equation}
is a singular solution of the boundary value problem \eqref{BVP}.
\end{proposition}

\begin{proof}
Since the coordinates $\lr{\ell,\psi,\theta}$ are orthogonal and the
scale factors $\abs{\nabla\ell}$ and $\abs{\nabla\psi}$ are equal, 
the vector field \eqref{P1BF} is a Beltrami field almost everywhere in $\Omega$.
Since such Beltrami field is endowed with the invariants $\theta$ and $L_{\theta}$, it also satisfy the boundary condition $\bol{w}\cdot\nabla\nu\lr{\theta,L_{\theta}}=0$ almost everywhere on $\p\Omega$.
Equations \eqref{p1eq1} and \eqref{p1eq2} guarantee that the Beltrami field is solenoidal almost everywhere in $\Omega$.
\end{proof}

Proposition \ref{prop1} can be combined with theorem \ref{thm1} 
to obtain the following result regarding the existence of singular solutions
to the boundary value problem \eqref{BVP}:

\begin{theorem}\label{thm2}
Let $\lr{\zeta,\beta,\gamma}\in C^{\infty}\lr{\Omega}$ be a smooth orthogonal coordinate system in a
smoothly bounded simply connected domain $\Omega\subset\mathbb{R}^3$ satisfying the hypothesis of theorem \ref{thm1}. Suppose that the unit outward normal $\bol{n}$ to the boundary $\p\Omega$ is such that $\bol{n}=\nabla\nu/\abs{\nabla\nu}$, where $\nu\in C^{\infty}\lr{\Omega'}$, $\Omega\subset\subset\Omega'$, is a smooth function of $\theta$ and $L_{\theta}$, $\nu=\nu\lr{\theta,L_{\theta}}$, $\lr{\ell,\theta,\psi}\in C^{\infty}\lr{\Omega'}$ the orthogonal coordinate system obtained from $\lr{\zeta,\beta,\gamma}$ according to the procedure of theorem \ref{thm1}, and $L_{\theta}=\ell\cos\theta-\psi\sin\theta$. Then, the vector field \eqref{solbf} is a family of singular solutions to the boundary value problem \eqref{BVP}.
\end{theorem}

In practice, theorem \ref{thm2} is applied as follows: first, given a smoothly bounded simply connected domain $\Omega$, one identifies the function $\nu$ whose gradient defines the normal direction to the bounding surface $\p\Omega$. Assuming $\nu=\gamma$, one completes the function $\gamma$ to an orthogonal set $\lr{\zeta,\beta,\gamma}$. If the topology of the domain $\Omega$ is compatible with the requirements of theorem \ref{thm1} on the metric coefficients of the coordinate system $\lr{\zeta,\beta,\gamma}$, singular solutions of the
boundary value problem \eqref{BVP} are available in the form of equation \eqref{solbf}. This approach was used to construct the spherical example discussed above.

We now return to the spherical example with $\ell=\phi$ and $\psi=\alpha\lr{\vartheta}$. Since the spherical radius $R$ is an invariant, $\bol{w}\cdot\nabla R=0$, a fluid element moving according to the spherical Beltrami flow \eqref{SB} will remain on the same spherical surface at all times $t\geq 0$.
On the surface $S_{0}=\{\bol{x}\in\mathbb{R}~\rvert~R=R_{0}>0\}$ one can identify the singular points $\bol{x}_{\pm}=\lr{0,0,\pm R_{0}}$.
At $\bol{x}_{\pm}$, which correspond to the angular values $\theta=\left\{0,\pi\right\}$, the modulus of $\bol{w}$ diverges:
\begin{equation}
\lim_{\bol{x}\rightarrow\bol{x}_{\pm}}\bol{w}^2=\lim_{\bol{x}\rightarrow\bol{x}_{\pm}}\abs{\nabla\psi}^2=\lim_{\bol{x}\rightarrow\bol{x}_{\pm}}\frac{\psi_{\alpha}^2+\psi^2_{\phi}}{R^2\sin^2\vartheta}=\lim_{\vartheta\rightarrow 0}\frac{1}{R^2_{0}\sin^2\vartheta}=+\infty.
\end{equation}  
Given a flow with velocity $\bol{w}$, it is useful to consider which initial conditions of the invariants $R$ and $L_{R}$ lead to the singular points $\bol{x}_{\pm}$. Observe that
\begin{equation}
L_{R}=\phi\cos R-\log\lr{\frac{\sin\vartheta}{1+\cos\vartheta}}\sin R.
\end{equation}
On the other hand,
\begin{equation}
\lim_{\vartheta\rightarrow 0}\log\lr{\frac{\sin\vartheta}{1+\cos\vartheta}}=-\infty,~~~~\lim_{\vartheta\rightarrow \pi}\log\lr{\frac{\sin\vartheta}{1+\cos\vartheta}}=+\infty.
\end{equation}
It follows that, whenever $\sin R\neq 0$, 
\begin{equation}
\lim_{\bol{x}\rightarrow\bol{x}_{\pm}}\abs{L_{R}}=+\infty.
\end{equation}
Therefore, the set of initial conditions allowing the singular points $\bol{x}_{\pm}$ has measure zero and is limited to an infinite value of the angular momentum-like quantity $L_{R}$. 
This result remains true for singularities of the variables $\ell$ and $\psi$ in general. Indeed, if the set of singularities of the variables $\ell$ and $\psi$ has measure zero, so is the measure of the corresponding set of values of the invariant $L_{\theta}=\ell\cos\theta-\psi\sin\theta$.

We conclude this section with an example where the procedure of theorem \ref{thm1} does not apply. Consider a toroidal coordinate system
\begin{equation}
\lr{\tau,\eta,\phi}=\lr{\sqrt{\lr{1-r}^2+z^2},\frac{z}{r-1},\phi}.
\end{equation}
Here, $\lr{r,\phi,z}$ is the cylindrical coordinate system introduced above, $\tau$ is the radial distance from the center of the torus located at $\lr{r,\phi,z}=\lr{1,\phi,0}$, and $\eta$ is an angle-like coordinate in the planes $\phi={\rm constant}$. 

Notice that such toroidal coordinates form an orthogonal set.
One has
\begin{equation}
\abs{\nabla\tau}^2=1,~~~~\abs{\nabla\eta}^2=\frac{\lr{1+\eta^2}^2}{\tau^2},~~~~\abs{\nabla\phi}^2=\frac{1}{\lr{1\pm\frac{\tau}{\sqrt{1+\eta^2}}}^2}.
\end{equation}
Here, the sign $\pm$ in the last formula depends on the sign of $r-1$, since $\lr{r-1}^2=\tau^2/\lr{1+\eta^2}$. 
Therefore, the only variable with the property that the corresponding scale factor is independent of the other coordinates is $\tau$. However, the ratio $\abs{\nabla\eta}/\abs{\nabla\phi}$ depends on the variables $\tau$ and $\eta$, and equation \eqref{hyp1} cannot be satisfied.

\section{Generalized Beltrami fields}

In this section we examine equation \eqref{Beq} when $\nabla P=\bol{0}$, i.e.
\begin{equation}
\BN{w}=\kappa\,\nabla\bol{w}^2~~~~{\rm in}~~\Omega.\label{Eq4_1}
\end{equation}
Observe that any solution of \eqref{Eq4_1} such that $\kappa\,\nabla\bol{w}^2\neq \bol{0}$ satisfies $\bol{w}\cdot\nabla\bol{w}^2=\cu\bol{w}\cdot\nabla\bol{w}^2=0$. Hence $\bol{w}^2$ is simultaneously an invariant of the flows generated by $\bol{w}$ and $\cu\bol{w}$. It will be shown that solving \eqref{Eq4_1} with $\kappa\neq 0$ is tantamount to finding a Beltrami field. For this reason and due to the fact that $\kappa=0$ corresponds to the standard Beltrami case, we shall refer to solutions of \eqref{Eq4_1} as generalized Beltrami fields. 
In the following, the notation $w=\abs{\bol{w}}$ is used.  

\begin{proposition}\label{propEq4_1}
Let $\bol{w}\in C^{\infty}\lr{\Omega}$ be a smooth vector field in a bounded domain $\Omega\subset\mathbb{R}^3$. Assume $\bol{w}\neq\bol{0}$ in $\Omega$ and $\kappa\neq 0$. Then, $\bol{w}$ is a solution of \eqref{Eq4_1} if and only if
\begin{equation}
\frac{\bol{w}}{w^{2\kappa}}\cp\left[{\cu\lr{\frac{\bol{w}}{w^{2\kappa}}}}\right]=\bol{0},~~~~\bol{w}\cdot\nabla w=0~~~~{\rm in}~~\Omega.\label{Eq4_2}
\end{equation}
\end{proposition}

\noindent The equivalence between \eqref{Eq4_1} and \eqref{Eq4_2} when $\kappa\neq 0$ can be verified by observing that
\begin{equation}
\frac{\bol{w}}{w^{2\kappa}}\cp\left[{\cu\lr{\frac{\bol{w}}{w^{2\kappa}}}}\right]=w^{-4\kappa}\left[{\BN{w}}-\kappa\,\hb{w}\cp\lr{\nabla w^{2}\cp\hb{w}}\right],
\end{equation}
where $\hb{w}=\bol{w}/w$. Whenever $\bol{w}\cdot\nabla w=0$, we have $\hb{w}\cp\lr{\nabla w^2\cp\hb{w}}=\nabla w^2$, and thus
\begin{equation}
\frac{\bol{w}}{w^{2\kappa}}\cp\left[{\cu\lr{\frac{\bol{w}}{w^{2\kappa}}}}\right]=w^{-4\kappa}\left[{\BN{w}}-\kappa\,\nabla w^{2}\right].
\end{equation}

This result enables one to apply the techniques for the construction of Beltrami fields developed in \cite{Sato} in order to obtain solutions of equation \eqref{Eq4_1}. Let us see how.
For $\kappa\neq 0$, we must distinguish two cases, $\kappa=1/2$ (this is the case of an ideal fluid), and $\kappa\neq 1/2$. 
Define $\bol{\xi}=\bol{w}/w^{2\kappa}$. 
We have $\xi=\abs{\bol{\xi}}=w^{1-2\kappa}$.
If $\kappa=1/2$, we see that $\xi=1$, implying that there is no functional relationship between $\xi$ and $w$. Hence, if one wants to obtain a solution of \eqref{Eq4_1} such that $w=\abs{\alpha}$, $\alpha\in C^{\infty}\lr{\Omega}$, it is sufficient to solve the following system in the domain $\Omega$,
\begin{subequations}\label{xi}
\begin{align}
\BN{\xi}&=\bol{0},\label{xia}\\
\abs{\bol{\xi}}&=1,\label{xib}\\
\nabla\alpha\cdot\bol{\xi}&=0.\label{xic}
\end{align}
\end{subequations}
Once $\bol{\xi}$ is obtained, the desired vector field is $\bol{w}=\alpha\,\bol{\xi}$. Indeed,
\begin{equation}
\BN{w}=\alpha^{2}\,\BN{\xi}+\alpha\,\bol{\xi}\cp\lr{\nabla\alpha\cp\bol{\xi}}=\frac{1}{2}\nabla\alpha^2=\frac{1}{2}\nabla w^2.
\end{equation}
A vector field $\bol{\xi}$ satisfying \eqref{xia} and \eqref{xib} 
can be constructed by application of corollary 1 of \cite{Sato}:
it is sufficient to determine an orthogonal coordinate system $\lr{\ell,\psi,\theta}$ such that $\abs{\nabla\ell}=\abs{\nabla\psi}=1$. Then, $\bol{\xi}=\cos{\theta}~\nabla\psi+\sin{\theta}~\nabla\ell$.
If the function $\alpha$ is not prescribed, solutions of \eqref{xic} can be obtained with the method of characteristics.   
If $\alpha$ is given, equation \eqref{xic} puts an additional constraint on the choice of $\lr{\ell,\psi,\theta}$. From corollary 1 of \cite{Sato} we know that, if \eqref{xic} is satisfied, then $\alpha$ is a function of $\theta$ and $L_{\theta}=\ell\,\cos{\theta}-\psi\,\sin\theta$, i.e. $\alpha=\alpha\lr{\theta,L_{\theta}}$.

When $\kappa\neq \left\{0,1/2\right\}$, solutions of \eqref{Eq4_1} can be obtained by solving the following system in $\Omega$,
\begin{subequations}\label{xi2}
\begin{align}
\BN{\xi}&=\bol{0},\label{xi2a}\\
\abs{\bol{\xi}}&=\abs{\alpha}^{1-2\kappa},\label{xi2b}\\
\nabla\alpha\cdot\bol{\xi}&=0.\label{xi2c}
\end{align}
\end{subequations}
Then, the solution is $\bol{w}=\alpha^{2\kappa}\,\bol{\xi}$.
From corollary 1 of \cite{Sato},
the vector field $\bol{\xi}$ can be constructed 
by determining an orthogonal coordinate system $\lr{\ell,\psi,\theta}$ 
such that $\abs{\nabla\ell}=\abs{\nabla\psi}=\abs{\alpha}^{1-2\kappa}$
and $\nabla\alpha\cdot\lr{\cos\theta\,\nabla\psi+\sin\theta\,\nabla\ell}=0$. 
As in the previous case, $\alpha$ will be a function of $\theta$ and $L_{\theta}$, i.e. $\alpha=\alpha\lr{\theta,L_{\theta}}$.

In the context of fluid dynamics and magnetohydrodynamics,
solenoidal solutions are desirable. 
For this purpose, systems \eqref{xi} and \eqref{xi2} must be modified
by adding the condition 
\begin{equation}
\di\bol{\xi}=0~~~~{\rm in}~~{\Omega}.
\end{equation} 
Since $\bol{w}=\alpha^{2\kappa}\,\bol{\xi}$ and $\nabla\alpha\cdot\bol{\xi}=0$, this is enough to ensure that $\di\bol{w}=0$.

We remark that any solution $\bol{w}$ to equation \eqref{Eq4_1} is
parallel to the Beltrami field $\bol{\xi}=\bol{w}/w^{2\kappa}$. Therefore it 
is endowed with the same integral invariants.
From \cite{Sato}, we know that any nontrivial Beltrami field (i.e. such that $h=\JI{\xi}\neq 0$ in $\Omega$) admits the local representation
$\bol{\xi}=\cos\theta~\nabla\psi+\sin\theta~\nabla\ell$ and two local invariants,
$\theta$ and $L_{\theta}$. Hence, the corresponding solution $\bol{w}$ also preserves $\theta$ and $L_{\theta}$ locally, i.e. $\bol{w}\cdot\nabla\theta=\bol{w}\cdot\nabla L_{\theta}=0$
in the neighborhood where the local representation holds.

\subsection{Examples}

In the remaining part of this section examples
of solutions to equation \eqref{Eq4_1} are given.

\begin{enumerate}

\item Let $\lr{x,y,z}$ be the standard Cartesian coordinate system.
Define $L_{z}=x\,\cos z-y\,\sin z$. Let $\alpha\in C^{\infty}\lr{\Omega}$ be an arbitrary function of $z$ and $L_{z}$, i.e. $\alpha=\alpha\lr{z,L_{z}}$.
The vector field
\begin{equation}
\bol{w}=\alpha\lr{z,L_{z}}\left\{\cos\left[ \sigma\lr{z}\right]~\nabla y+\sin\left[\sigma\lr{ z}\right]~\nabla x\right\},\label{GBxyz}
\end{equation}
is a solution of \eqref{Eq4_1} for $\kappa=1/2$ such that $w=\abs{\alpha}$.
Here, $\sigma\lr{z}$ is an arbitrary $C^{\infty}\lr{\Omega}$ function of $z$. 
Moreover $\nabla\cdot\bol{w}=0$. 
The flow generated by $\bol{w}$ admits the invariants $z$ and $L_{z}$.
A plot of \eqref{GBxyz} is given in figure \ref{fig2}.

\begin{figure}[h]
\hspace*{-0cm}\centering
\includegraphics[scale=0.3]{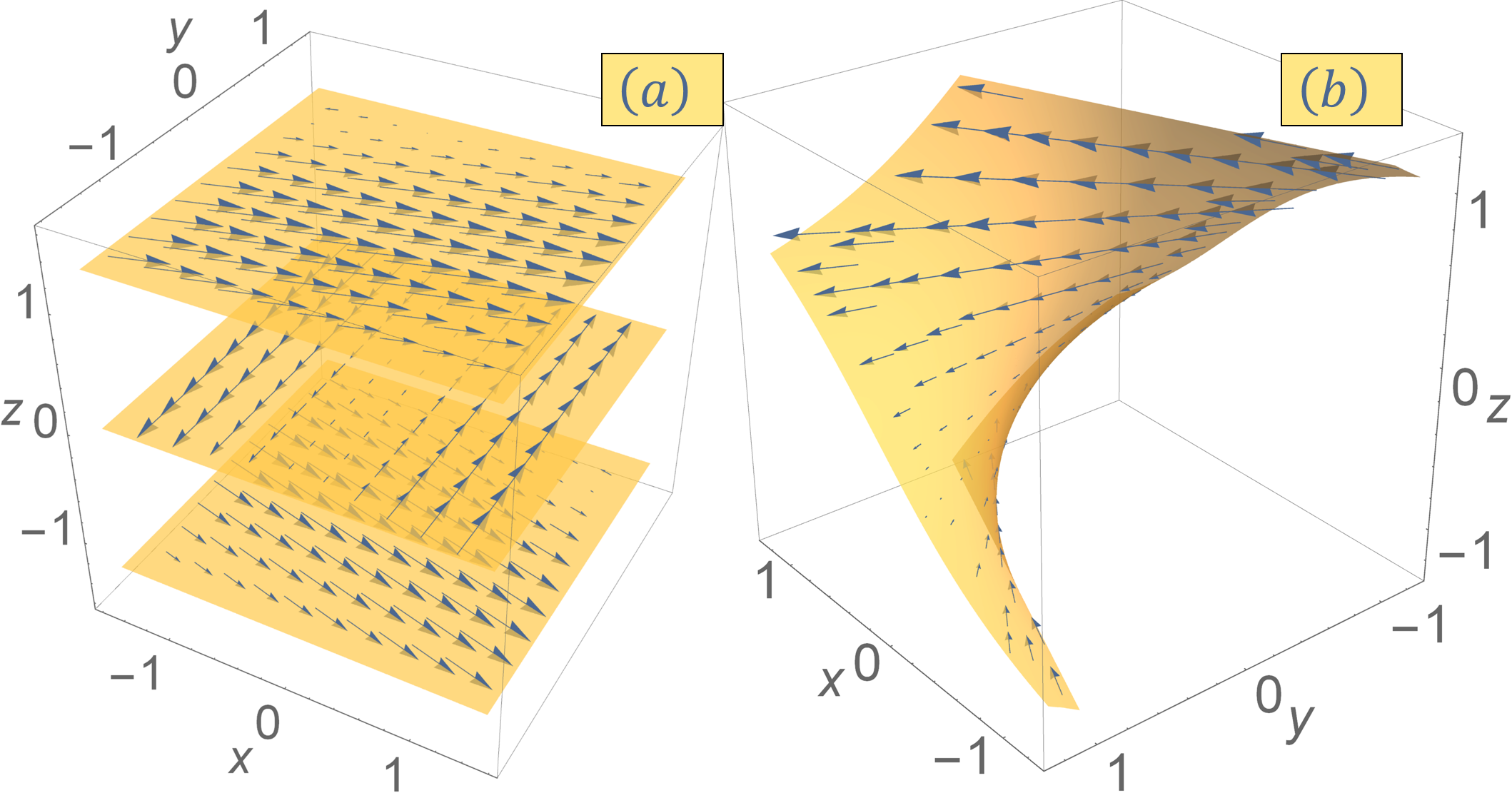}
\caption{\footnotesize (a): Plot of the generalized Beltrami field \eqref{GBxyz} on the surfaces $z=-1.3,0,1.3$. (b): Plot of the generalized Beltrami field \eqref{GBxyz} on the surface $L_{z}=0.5$. Here $\sigma\lr{z}=z$ and $\alpha\lr{z,L_{z}}=\sin\lr{z+L_{z}}$.}
\label{fig2}
\end{figure}

\item Let $\lr{r,\phi,z}$ be the cylindrical coordinate system introduced in section 3. Define $L_{\phi}=r\,\cos \phi-z\,\sin \phi$. Let $\alpha\in C^{\infty}\lr{\Omega}$ be an arbitrary function of $\phi$ and $L_{\phi}$, i.e. $\alpha=\alpha\lr{\phi,L_{\phi}}$.
The vector field
\begin{equation}
\bol{w}=\alpha\lr{\phi,L_{\phi}}\left\{\cos\left[\sigma\lr{ \phi}\right]~\nabla z+\sin\left[\sigma\lr{ \phi}\right]~\nabla r\right\},\label{GBrphiz}
\end{equation}
is a solution to \eqref{Eq4_1} for $\kappa=1/2$ such that $w=\abs{\alpha}$.
Here, $\sigma\lr{\phi}$ is an arbitrary $C^{\infty}\lr{\Omega}$ function of $\phi$. 
Moreover $\nabla\cdot\bol{w}=\alpha\sin\left[\sigma\lr{\phi}\right]/r$. 
The flow generated by $\bol{w}$ admits the invariants $\phi$ and $L_{\phi}$.
A plot of \eqref{GBrphiz} is given in figure \ref{fig3}.

\begin{figure}[h]
\hspace*{-0cm}\centering
\includegraphics[scale=0.3]{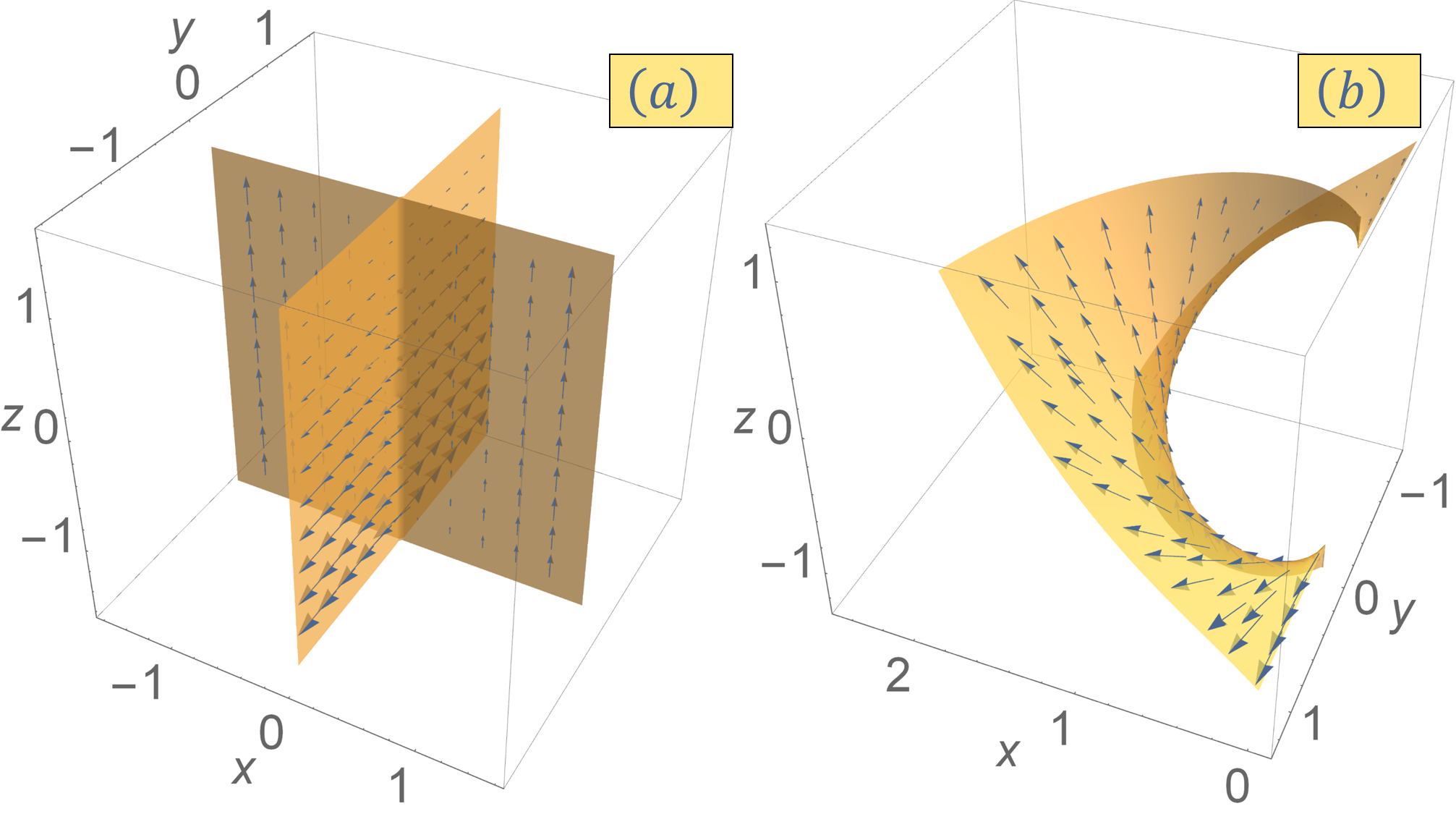}
\caption{\footnotesize (a): Plot of the generalized Beltrami field \eqref{GBrphiz} on the surfaces $\phi=0,\pi$. (b): Plot of the generalized Beltrami field \eqref{GBrphiz} on the surface $L_{\phi}=1$. Here $\sigma\lr{\phi}=\phi$ and $\alpha\lr{\phi,L_{\phi}}=\sin\lr{\phi+L_{\phi}}$.}
\label{fig3}
\end{figure}

\end{enumerate}

\section{Ideal MHD equilibria}

This section is concerned with equation \eqref{Beq} when $\kappa=0$,
\begin{equation}
\BN{w}=\nabla P~~~~{\rm in}~~\Omega.\label{Eq2}
\end{equation}
Local conditions for existence of solutions are examined and
families of analytic solutions are given. 
Solenoidal solutions, i.e. ideal MHD equilibria, are discussed
in the second part of this section.

Observe that any solution of \eqref{Eq2} for $\nabla P\neq\bol{0}$ satisfies $\bol{w}\cdot\nabla P=\cu\bol{w}\cdot\nabla P=0$. Hence $P$ is simultaneously an invariant of the flows generated by $\bol{w}$ and $\cu\bol{w}$. 
The following theorem provides local conditions for the existence of solutions to system \eqref{Eq2}. 

\begin{proposition}\label{prop3}
Let $\bol{w}\in C^{\infty}\lr{\Omega}$ be a smooth vector field and 
$P\in C^{\infty}\lr{\Omega}$ a smooth function in a bounded domain $\Omega\subset\mathbb{R}^3$ such that $\nabla P\neq\bol{0}$ in $\Omega$.
Then, $\bol{w}$ is a solution of equation \eqref{Eq2} if and only if
for every $\bol{x}\in\Omega$ there exist a neighborhood $U\subset\Omega$ of $\bol{x}$ and local functions $\lr{\mu,\lambda,C}\in C^{\infty}\lr{U}$ such that
\begin{subequations}\label{Eq2s}
\begin{align}
\lambda&=\lambda\lr{P,C},\label{Eq2a}\\
0&=\frac{1}{2}\frac{\p\lambda^2}{\p P}\abs{\nabla C}^2+\frac{\p\lambda}{\p P}\nabla\mu\cdot\nabla C-1,\label{Eq2b}\\
0&=\nabla\mu\cdot\nabla P+\lambda\nabla C\cdot\nabla P,\label{Eq2c}
\end{align}
\end{subequations} 
and
\begin{equation}
\bol{w}=\nabla\mu+\lambda\,\nabla C~~~~{\rm in}~~U.\label{wthm2}
\end{equation}
\end{proposition}

\begin{proof}
First, we prove that system \eqref{Eq2} implies system \eqref{Eq2s} and equation \eqref{wthm2}. Observe that $\nabla P\neq\bol{0}$ can be satisfied only if both $\bol{w}\neq\bol{0}$ and $\cu\bol{w}\neq\bol{0}$ in $\Omega$. 
The latter condition guarantees that the hypothesis of the Lie-Darboux theorem are verified (see theorem 1 of \cite{Sato} and references \cite{DeLeon89,Arnold89,Salmon17}). Hence, for every $\bol{x}\in\Omega$ 
there exist a neighborhood $U\subset\Omega$ of $\bol{x}$ and local functions $\lr{\mu,\lambda,C}\in C^{\infty}\lr{U}$ such that
\begin{equation}
\bol{w}=\nabla\mu+\lambda\,\nabla C~~~~{\rm in}~~U.
\end{equation}
Furthermore, from equation \eqref{Eq2} we have
\begin{equation}
\nabla P\cdot\cu\bol{w}=\nabla P\cdot\nabla\lambda\cp\nabla C=0~~~~{\rm in}~~U.
\end{equation}
Hence, $P=P\lr{\lambda,C}$ in $U$. Since $\nabla P\neq\bol{0}$, $P_{\lambda}$ and $P_{C}$ cannot vanish simultaneously.  
Suppose that $P_{\lambda}=0$ and $P_{C}\neq 0$. Then, the reparametrization
$\mu^{'}=\mu+\lambda\, C$, $\lambda^{'}=-C$, and $C^{'}=\lambda$ gives $P_{\lambda^{'}}=-P_{C}\neq 0$. Therefore, we can assume without loss of generality that $P_{\lambda}\neq 0$ in $U$. By the inverse function theorem, the function $P$ is locally invertible to $\lambda=\lambda\lr{P,C}$. Again, we can assume $U$ to be the domain of such inverse function.  

Next, we evaluate $\BN{w}$. After some manipulations
\begin{equation}
\begin{split}
\bol{w}\cp &\lr{\cu\bol{w}}=\lr{\nabla\mu\cdot\nabla C+\lambda\,\abs{\nabla C}^2}\lambda_{P}\,\nabla P\\&-\left(\nabla\mu\cdot\nabla P+\lambda\nabla C\cdot\nabla P\right)\lambda_{P}\,\nabla C.\label{wcurlw}
\end{split}
\end{equation}
On the other hand $\BN{w}=\nabla P$. Since $\cu\bol{w}\neq\bol{0}$, we must also have $\nabla P\cp\nabla C\neq\bol{0}$, meaning that the directions $\nabla P$ and $\nabla C$ are independent. It follows that the factor multiplying on $\nabla P$ in equation \eqref{wcurlw} must equal $1$, while the factor on $\nabla C$ must equal $0$.
We have thus obtained system \eqref{Eq2s} and the first implication is proven.
The proof of the converse statement is trivial.
\end{proof}

Observe that the local condition for $\bol{w}$ to be solenoidal is
\begin{equation}
\di\bol{w}=\Delta\mu+\lambda_{P}\nabla P\cdot\nabla C+\lambda_{C}\abs{\nabla C}^2+\lambda\,\Delta C=0~~~~{\rm in}~~U.
\end{equation}
Proposition \ref{prop3} suggests the following method to solve equation \eqref{Eq2}. Suppose that, for a given $P$, one can
construct an orthogonal set $\lr{\mu,P,C}$. Then, equation 
\eqref{Eq2c} is automatically satisfied, and the desired solution can be obtained
if one can find a function $\lambda$ with the property, arising from equation \eqref{Eq2b}, $\frac{\p\lambda^2}{\p P}=2/{\abs{\nabla C}^2}$.  
On this regard, we have the following.

\begin{proposition}\label{Eq2_2}
Let $\lr{\mu,P,C}\in C^{\infty}\lr{\Omega}$ be an orthogonal coordinate system in a bounded domain $\Omega$ such that 
\begin{equation}
\frac{\p}{\p\mu}\abs{\nabla C}^2=0.
\end{equation}
Then, the vector field
\begin{equation}
\bol{w}=\nabla \sigma\lr{\mu}+\sqrt{2\int{\frac{dP}{\abs{\nabla C}^2}}}~\nabla C,
\end{equation}
is a solution of \eqref{Eq2}. 
Here, $\sigma\lr{\mu}$ is an arbitrary $C^{\infty}\lr{\Omega}$ function of $\mu$.
\end{proposition}

Proposition \ref{Eq2_2} can be verified by  
using the orthogonality of the coordinate system $\lr{\mu,P,C}$,
\begin{equation}
\begin{split}
\bol{w}\cp \lr{\cu\bol{w}}=&\sigma_{\mu}\nabla\mu\cp\lr{\nabla\sqrt{2\int{\frac{dP}{\abs{\nabla C}^2}}}\cp\nabla C}\\
&+\nabla C\cp\lr{\nabla\int{\frac{dP}{\abs{\nabla C}^2}}\cp\nabla C}\\=&\nabla P.
\end{split}
\end{equation}
Hence, for a given $P$, a solution $\bol{w}=\nabla\sigma\lr{\mu}+\lambda\,\nabla C$ of \eqref{Eq2} can be obtained if an orthogonal coordinate system $\lr{\mu,P,C}$ can be found such that $\p_{\mu}\abs{\nabla C}^2=0$. Then, $\lambda=\sqrt{2\int{\frac{dP}{\abs{\nabla C}^2}}}$. Below we give some examples of application of proposition \ref{Eq2_2}.

\subsection{Examples}

\begin{enumerate}

\item Let $\lr{x,y,z}$ be the standard Cartesian coordinate system.
Let $P\in C^{\infty}\lr{\Omega}$ be a function of the spatial coordinate $x$, $P=P\lr{x}$, and $\sigma\in C^{\infty}\lr{\Omega}$ an arbitrary function of the coordinate $z$, $\sigma=\sigma\lr{z}$. Then, the vector field
\begin{equation}
\bol{w}=\nabla \sigma\lr{z}+\sqrt{2 P}~\nabla y,
\end{equation}   
is a solution of equation \eqref{Eq2}. Furthermore $\di\bol{w}=\sigma_{zz}$, which vanishes for $\sigma_{zz}=0$. 
In a similar manner, one can construct solutions for given profiles
$P=P\lr{y}$ and $P=P\lr{z}$. Next, consider the case $P=P\lr{x,y}$. The first order partial differential equation for the variable $C=C\lr{x,y}$,
\begin{equation}
\nabla P\cdot\nabla C=0,
\end{equation}
can be solved, at least locally, by application of the method of characteristics.
Note that, by construction, $\p_{z}\abs{\nabla C}=0$. Hence, the hypothesis of proposition \ref{Eq2_2} are satisfied.
Thus, the vector field
\begin{equation}
\bol{w}=\nabla \sigma\lr{z}+\sqrt{2\int{\frac{dP}{\abs{\nabla C}^2}}}~\nabla C,
\end{equation}
is a local solution of equation \eqref{Eq2}.
For example, assume $P=xy$. We find $C=\frac{y^2-x^2}{2}$ and
\begin{equation}
\bol{w}=\nabla \sigma\lr{z}+\sqrt{\log\lr{\sqrt{C^2+P^2}+P}}~\nabla C.
\end{equation}
Let $P=e^{x}\sin{y}$. We have $C=e^{x}\cos{y}$ and
\begin{equation}
\bol{w}=\nabla \sigma\lr{z}+\sqrt{\frac{2}{C}\arctan\lr{\frac{P}{C}}}~\nabla C.
\end{equation} 
The cases $P=P\lr{x,z}$ and $P=P\lr{y,z}$ can be treated 
with the same approach. 

\item Consider cylindrical coordinates $\lr{r,\phi,z}$. Let $P\in C^{\infty}\lr{\Omega}$ be a function of the radial coordinate $r$, $P=P\lr{r}$. Let $\sigma\in C^{\infty}\lr{\Omega}$ be an arbitrary function of the angle variable $\phi$, $\sigma=\sigma\lr{\phi}$. Then, the vector field
\begin{equation}
\bol{w}=\nabla \sigma\lr{\phi}+\sqrt{2P}~\nabla z,
\end{equation}
is a solution of equation \eqref{Eq2}. Furthermore $\di\bol{w}=\sigma_{\phi\phi}/r^2$, which vanishes for $\sigma_{\phi\phi}=0$. In a similar manner, one can construct solutions for given profiles $P=P\lr{\phi}$ and $P=P\lr{z}$.
Next, consider the case $P=P\lr{r,z}$. The first order partial differential equation for the variable $C=C\lr{r,z}$, 
\begin{equation}
\nabla P\cdot\nabla C=0,
\end{equation}
can be solved, at least locally, by application of the method of characteristics. 
Note that, by construction, $\p_{\phi}\abs{\nabla C}=0$. Hence, the hypothesis of proposition \ref{Eq2_2} are satisfied.
Thus, the vector field
\begin{equation}
\bol{w}=\nabla \sigma\lr{\phi}+\sqrt{2\int{\frac{dP}{\abs{\nabla C}^2}}}~\nabla C,
\end{equation}
is a local solution of equation \eqref{Eq2}. For example, assume
$P=\exp\lr{-r-z}$. We find $C=r-z$ and
\begin{equation}
\bol{w}=\nabla \sigma\lr{\phi}+\sqrt{\,\exp\lr{-r-z}}~\nabla \lr{r-z}.
\end{equation} 
The case $P=P\lr{r,\phi}$ can be treated in a similar way.
\end{enumerate}

\subsection{Local representation of ideal MHD equilibria}

Ideal MHD equilibria correspond to solenoidal solutions of \eqref{Beq}
for $\kappa=0$: 
\begin{equation}\label{IMHD}
\BN{w}=\nabla P,~~~~\di\bol{w}=0~~~~{\rm in} ~~~~{\Omega}.
\end{equation}
The divergence free requirement poses a stringent condition on the set of solutions. When a solution is assumed to possess Euclidean symmetries, 
i.e. symmetry under translation, rotation, or their combination, 
system \eqref{IMHD} reduces to a single elliptic partial differential equation,
the Grad-Shafranov equation \cite{Grad58,Edenstrasser80_1,Edenstrasser80_2}. 
Here, we show that general ideal MHD equilibria are locally described
by a system of two coupled second order partial differential equations,  
and derive local conditions for existence of solutions.

\begin{proposition}\label{prop5}
Let $\bol{w}\in C^{\infty}\lr{\Omega}$ be a smooth vector field in a 
bounded domain $\Omega\subset\mathbb{R}^3$. Assume $\bol{w}\neq\bol{0}$ in $\Omega$. 
Then, for every $\bol{x}\in\Omega$ there exist a neighborhood $U\subset\Omega$ of $\bol{x}$ and local smooth functions $\lr{\chi,\pi}\in C^{\infty}\lr{U}$ such that 
\begin{equation}
\bol{w}=\nabla \chi\cp\nabla\pi~~~~{\rm in}~~U.
\end{equation}
Furthermore, system \eqref{IMHD} is locally equivalent to
\begin{subequations}\label{IMHD2}
\begin{align}
\di\left[{\nabla\pi\cp\lr{\nabla \chi\cp\nabla\pi}}\right]&=P_{\chi},\\
\di\left[{\nabla \chi\cp\lr{\nabla\pi\cp\nabla \chi}}\right]&=P_{\pi}~~~~{\rm in}~~U.
\end{align}
\end{subequations}
\end{proposition}
\begin{proof}
Since $\bol{w}\neq\bol{0}$ and $\di\bol{w}=0$, the Lie-Darboux theorem of differential geometry guarantees that there exists local smooth functions $\lr{\chi,\pi}\in C^{\infty}\lr{U}$ such that $\bol{w}=\nabla\chi\cp\nabla\pi$. From $\bol{w}\cdot\nabla P=0$, we have $P=P\lr{\chi,\pi}$.
Then, in $U$, 
\begin{equation}\label{p5_1}
\begin{split}
\bol{w}&\cp\lr{\cu\bol{w}}=\\
&=\bol{w}\cp\left[{\Delta\pi\,\nabla\chi-\Delta\chi\,\nabla\pi+\lr{\nabla\pi\cdot\nabla}\nabla\chi-\lr{\nabla\chi\cdot\nabla}\nabla\pi}\right]
\\&=\Delta\pi\,\nabla\chi\cp\lr{\nabla\pi\cp\nabla\chi}
+\Delta\chi\,\nabla\pi\cp\lr{\nabla\chi\cp\nabla\pi}
\\&~~~~+\left[\lr{\nabla\pi\cdot\nabla}\nabla\chi\right]\cp\lr{\nabla\pi\cp\nabla\chi}
+\left[\lr{\nabla\chi\cdot\nabla}\nabla\pi\right]\cp\lr{\nabla\chi\cp\nabla\pi}
\\&=\left(-\Delta\pi\,\nabla\chi\cdot\nabla\pi+\Delta\chi\,\abs{\nabla\pi}^2+\nabla\pi\cdot\lr{\nabla\chi\cdot\nabla}\nabla\pi\right)\nabla\chi
\\&~~~~-\lr{\nabla\pi\cdot\lr{\nabla\pi\cdot\nabla}\nabla\chi}\nabla\chi
\\&~~~~+\lr{-\Delta\chi\,\nabla\pi\cdot\nabla\chi+\Delta\pi\,\abs{\nabla\chi}^2+\nabla\chi\cdot\lr{\nabla\pi\cdot\nabla}\nabla\chi}\nabla\pi
\\&~~~~-\left[{\nabla\chi\cdot\lr{\nabla\chi\cdot\nabla}\nabla\pi}\right]\nabla\pi
\\&=\left\{{\di\left[{\abs{\nabla\pi}^2\nabla\chi-\lr{\nabla\pi\cdot\nabla\chi}\nabla\pi}\right]}+
\lr{\nabla\pi\cdot\nabla}\lr{\nabla\pi\cdot\nabla\chi}\right\}\nabla\chi
\\&~~~~+\left\{-\lr{\nabla\chi\cdot\nabla}\abs{\nabla\pi}^2+\nabla\pi\cdot\left[{\lr{\nabla\chi\cdot\nabla}\nabla\pi-\lr{\nabla\pi\cdot\nabla}\nabla\chi}\right]\right\}\nabla\chi
\\&~~~~+\left\{{\di\left[{\abs{\nabla\chi}^2\nabla\pi-\lr{\nabla\chi\cdot\nabla\pi}\nabla\chi}\right]}+
\lr{\nabla\chi\cdot\nabla}\lr{\nabla\chi\cdot\nabla\pi}\right\}\nabla\pi
\\&~~~~+\left\{-\lr{\nabla\pi\cdot\nabla}\abs{\nabla\chi}^2+\nabla\chi\cdot\left[{\lr{\nabla\pi\cdot\nabla}\nabla\chi-\lr{\nabla\chi\cdot\nabla}\nabla\pi}\right]\right\}\nabla\pi.
\\&=\left\{\di\left[\nabla\pi\cp\lr{\nabla\chi\cp\nabla\pi}\right]\right\}\nabla\chi
+\left\{\di\left[\nabla\chi\cp\lr{\nabla\pi\cp\nabla\chi}\right]\right\}\nabla\pi.
\end{split}
\end{equation}
On the other hand,
\begin{equation}\label{p5_2}
\BN{w}=P_{\chi}\,\nabla\chi+P_{\pi}\,\nabla\pi~~~~{\rm in}~~U.
\end{equation}
Combining \eqref{p5_1} and \eqref{p5_2}, one obtains \eqref{IMHD2}.
\end{proof}

The variable $\chi$ above can be identified with $P$:

\begin{proposition}\label{prop6}
Let $\bol{w}\in C^{\infty}\lr{\Omega}$ be a smooth vector field in a 
bounded domain $\Omega\subset\mathbb{R}^3$. Assume $\bol{w}\neq\bol{0}$ and $\nabla P\neq\bol{0}$ in $\Omega$. 
Then, for every $\bol{x}\in\Omega$ there exist a neighborhood $U\subset\Omega$ of $\bol{x}$ and a local smooth function $\zeta\in C^{\infty}\lr{U}$ such that 
\begin{equation}
\bol{w}=\nabla P\cp\nabla\zeta~~~~{\rm in}~~U.
\end{equation}
Furthermore, system \eqref{IMHD} is locally equivalent to
\begin{subequations}\label{IMHD3}
\begin{align}
\di\left[{\nabla\zeta\cp\lr{\nabla P\cp\nabla\zeta}}\right]&=1,\\
\di\left[{\nabla P\cp\lr{\nabla\zeta\cp\nabla P}}\right]&=0~~~~{\rm in}~~U.
\end{align}
\end{subequations}
\end{proposition}
\begin{proof}
We need to show that locally $\bol{w}=\nabla P\cp\nabla\zeta$.
From proposition \ref{prop5} we have $\bol{w}=\nabla\chi\cp\nabla\pi$ in some neighborhood $U$.  From $\nabla P\cdot\bol{w}=0$, it follows that $P=P\lr{\chi,\pi}$ in $U$. 
Furthermore, since $\nabla P\neq\bol{0}$, from the inverse function theorem the function $P$ can be inverted locally to obtain either $\pi=\pi\lr{P,\chi}$ or $\chi=\chi\lr{P,\pi}$. Suppose that $\chi=\chi\lr{P,\pi}$. Then
\begin{equation}
\bol{w}=\chi_{P}\,\nabla P\cp\nabla\pi=\nabla P\cp\nabla\lr{\int{\chi_{P}}\,d\pi}.
\end{equation}
Thus $\zeta=\int{\chi_{P}}\,d\pi$. 
Equation \eqref{IMHD3} follows by the same argument used in proposition \ref{prop5}.
\end{proof}

\section{General steady Euler flows}
In this section we study the local solvability of equation \eqref{Beq}, 
\begin{equation}
\BN{w}=\nabla\lr{P+\kappa\,\bol{w}^2}~~~~{\rm in}~~\Omega.\label{Eq3}
\end{equation}
We define 
\begin{equation}
\mc{P}=P+\kappa\,\bol{w}^2.
\end{equation}
The following proposition casts the steady equation \eqref{Eq3} 
into a time-dependent Hamiltonian system for the local coordinates
spanning the local solution $\bol{w}$. 

\begin{proposition}\label{prop7}
Let $\bol{w}\in C^{\infty}\lr{\Omega}$ be a smooth vector field in a bounded domain $\Omega\subset\mathbb{R}^3$ with $h=\bol{w}\cdot\cu\bol{w}\neq 0$ in $\Omega$.
Then $\bol{w}$ satisfies equation \eqref{Eq3} if and only if for every $\bol{x}\in\Omega$ there exists a neighborhood $U\subset\Omega$ of $\bol{x}$ and local coordinates $\lr{\mu,\lambda,C}\in C^{\infty}\lr{\Omega}$ such that 
\begin{subequations}\label{EulerH}
\begin{align}
\dot{\lambda}&=\bol{w}\cdot\nabla\lambda=-\frac{\p \mc{P}}{\p C},\\
\dot{C}&=\bol{w}\cdot\nabla C=\frac{\p \mc{P}}{\p \lambda},\\
0&=\frac{\p \mc{P}}{\p\mu},
\end{align}
\end{subequations}
and
\begin{equation}
\bol{w}=\nabla\mu+\lambda\,\nabla C~~~~{\rm in}~~U.\label{locrep}
\end{equation}
\end{proposition}
Here, the notation $\dot{f}=\bol{w}\cdot\nabla f$ was introduced to 
denote the rate of change of a function $f\in C^{\infty}\lr{\Omega}$ 
along the flow generated by $\bol{w}$.

\begin{proof}
First we prove that \eqref{Eq3} implies \eqref{EulerH} and \eqref{locrep}.
Equation \eqref{locrep} is a consequence of the Lie-Darboux theorem
(see theorem 1 of \cite{Sato}). The condition $h\neq 0$ ensures that 
the Jacobian of the change of variables $\lr{\mu,\lambda,C}\mapsto\lr{x,y,z}$ does not vanish,
making $\lr{\mu,\lambda,C}$ a coordinate system in $U$.
Substituting \eqref{locrep} in equation \eqref{Eq3}, we have
\begin{equation}
\bol{w}\cp\lr{\nabla\lambda\cp\nabla C}=\nabla\mc{P}~~~~{\rm in}~~U.\label{Euler2}
\end{equation}
Next, observe that
\begin{equation}
h=\JI{w}=\nabla\mu\cdot\nabla\lambda\cp\nabla C\neq 0~~~~{\rm in}~~U.
\end{equation}
Hence, the tangent basis $\lr{\p_{\mu},\p_{\lambda},\p_{C}}$ can be expressed as
\begin{equation}
\p_{\mu}=h^{-1}\nabla\lambda\cp\nabla C,~~~~
\p_{\lambda}=h^{-1}\nabla C\cp\nabla \mu,~~~~
\p_{C}=h^{-1}\nabla\mu\cp\nabla C.
\end{equation}
Then, equation \eqref{Euler2} becomes
\begin{equation}
h\,\bol{w}\cp\p_{\mu}=\nabla\mc{P}~~~~{\rm in}~~U.
\end{equation}
Projecting this equation on the tangent basis $\lr{\p_{\mu},\p_{\lambda},\p_{C}}$
we obtain system \eqref{EulerH}. This proves the first implication.
The proof of the converse statement is trivial.
\end{proof}
Note that system \eqref{EulerH} has a Hamiltonian form 
with Hamiltonian $\mc{P}=\mc{P}\lr{\lambda,C}$ and phase space coordinates $\lr{p,q}=\lr{\lambda,C}$.
A similar result applies even if $h$ is allowed to vanish but $\cu\bol{w}\neq\bol{0}$ in $\Omega$: 
since $\nabla\mc{P}\cdot\cu\bol{w}=0$, we still have $\mc{P}=\mc{P}\lr{\lambda,C}$ in $U$. Then,  
\begin{equation}
\begin{split}
\BN{w}&=\bol{w}\cp\lr{\nabla\lambda\cp\nabla C}\\&=\lr{\bol{w}\cdot\nabla C}\,\nabla\lambda-\lr{\bol{w}\cdot\nabla\lambda}\,\nabla C=\mc{P}_{\lambda}\,\nabla\lambda+\mc{P}_{C}\,\nabla C.
\end{split}
\end{equation}
This gives an Hamiltonian system $\dot{\lambda}=-\mc{P}_{C}$, $\dot{C}=\mc{P}_{\lambda}$. However, in this case $\lr{\mu,\lambda,C}$ does not represent a coordinate system in $U$. 
The case, $\cu\bol{w}=\bol{0}$ in $\Omega$ is trivial, since it implies $\bol{w}=\nabla f$ for some $f\in C^{\infty}\lr{\Omega}$, and $\mc{P}={\rm constant}$. 
Again, proposition \ref{prop7} and the resulting local representation for $\bol{w}$ suggests the use of orthogonal coordinates to construct nontrivial solutions of \eqref{Eq3}. Indeed, observe that, if the variables $\lr{\mu,\lambda,C}$ are mutually orthogonal, system \eqref{EulerH} reduces to
\begin{subequations}
\begin{align}
0&=-\frac{\p\mc{P}}{\p C},\\
\lambda\abs{\nabla C}^2&=\frac{\p\mc{P}}{\p\lambda},\\
0&=\frac{\p\mc{P}}{\p\mu}.
\end{align}
\end{subequations}
We have the following:

\begin{proposition}\label{ConstrGen}
Let $\lr{\mu,\lambda,C}\in C^{\infty}\lr{\Omega}$ be an orthogonal
coordinate system such that, in $\Omega$,
\begin{subequations}
\begin{align}
\frac{1}{2}\abs{\nabla C}^2&=\frac{\p f}{\p\lambda^2},\label{dC}\\
\frac{1}{2}\abs{\nabla\mu}^2&=-P+f-\lambda^2\frac{\p f}{\p\lambda^2}+c,\label{dmu}
\end{align}
\end{subequations}
where $f\in C^{\infty}\lr{\Omega}$ is a function of $\lambda^2$, i.e. $f=f\lr{\lambda^2}$, and $c\in\mathbb{R}$ a real constant. Then, the vector field
\begin{equation}
\bol{w}=\nabla\mu+\lambda\,\nabla C,\label{w2}
\end{equation}
is a solution of \eqref{Eq3}. Furthermore, $\lambda$ is an invariant of the flow generated by \eqref{w2}.
\end{proposition}

\begin{proof}
Since $\lr{\mu,\lambda,C}$ is orthogonal,
\begin{equation}
\bol{w}^2=
\abs{\nabla\mu}^2+2\lambda^2\frac{\p f}{\p\lambda^2}.
\end{equation}
Here we used equation \eqref{dC}.
From equation \eqref{dmu}, we thus have
\begin{equation}
\nabla f=\nabla\lr{P+\frac{1}{2}\bol{w}^2}.\label{df}
\end{equation}
On the other hand,
\begin{equation}
\bol{w}\cp\lr{\cu\bol{w}}=\frac{1}{2}\abs{\nabla C}^2\nabla\lambda^2=\nabla f.\label{wxdw}
\end{equation}
In the last passage we used equation \eqref{dC}. 
Combining equations \eqref{df} and \eqref{wxdw} gives the desired result.
Finally, $\lambda$ is an invariant of the flow generated by \eqref{w2} because the coordinate system $\lr{\mu,\lambda,C}$ is orthogonal by construction.
\end{proof}

Proposition \ref{ConstrGen} can be used to construct classes of solutions
with certain topologies:

\begin{proposition}\label{Construction}
Let $\lr{\alpha,\beta,\gamma}\in C^{\infty}\lr{\Omega}$ be an orthogonal coordinate system such that
$\abs{\nabla\alpha}=k$, with $k$ a positive real constant, and $\p_{\alpha}\abs{\nabla\beta}=\p_{\alpha}\abs{\nabla\gamma}=0$. 
Then, for every $\bol{x}\in\Omega$ and for every $P\in C^{\infty}\lr{\ov{\Omega}}$ with $P=P\lr{\beta,\gamma}$ a smooth function of $\beta$ and $\gamma$, there exist a neighborhood $U\subset\Omega$ of $\bol{x}$, a positive real constant $c$, and local coordinates $\lr{\mu,\lambda,C}\in C^{\infty}\lr{U}$ satisfying
\begin{subequations}
\begin{align}
\frac{1}{2}\abs{\nabla\mu}^2&=c-P,\\
\nabla\mu\cdot\nabla\lambda&=0,\label{dmudlambda}\\
C&=\sqrt{2}\,\frac{\alpha}{k},
\end{align}
\end{subequations} 
such that the vector field
\begin{equation}
\bol{w}=\nabla\mu+\lambda\,\nabla C~~~~{\rm in}~~U,\label{wloc}
\end{equation}
is a local solution of \eqref{Eq3}. Furthermore $\lambda$ is a local invariant of the flow generated by \eqref{wloc}. 
\end{proposition}

\begin{proof}
First observe that, since $P$ is continuous in a closed bounded interval, it is bounded by some positive real constant $c$, i.e. $P< c$ in $\Omega$.
Now we apply proposition \eqref{ConstrGen} by setting $f=\lambda^2$.
Then equation \eqref{dC} can be satisfied by choosing $C=\sqrt{2}\,\alpha/k$.
Next, we look for a solution $\mu=\mu\lr{\beta,\gamma}$ of equation \eqref{dmu}.
Such solution exists, at least locally, because $\abs{\nabla\mu}$ does not depend on the variable $\alpha$, and the method of characteristics can be applied. Furthermore, it satisfies $\abs{\nabla\mu}^2=2\lr{c-P}$.
Once $\mu$ is obtained, we can apply again the method of characteristics
to calculate $\lambda\lr{\beta,\gamma}$ from equation \eqref{dmudlambda}. Then the vector field \eqref{wloc} satisfies equation \eqref{Eq3} in some appropriate neighborhood $U\subset\Omega$ around any point $\bol{x}\in\Omega$. Furthermore, since the cotangent vectors $\lr{\nabla\mu,\nabla\lambda,\nabla C}$ are orthogonal by construction, $\bol{w}\cdot\nabla\lambda=0$. Thus $\lambda$ is a local invariant of the flow generated by \eqref{wloc}. Finally, the variables $\lr{\mu,\lambda,C}$ form a local coordinate system such that the Jacobian of the coordinate transformation is
\begin{equation}
h=\JI{w}=\nabla\mu\cdot\nabla\lambda\cp\nabla C=\sqrt{c-P}\,\abs{\nabla\lambda}~~~~{\rm in}~~U.
\end{equation}
Here, we assumed that the orientation of the coordinate system $\lr{\mu,\lambda,C}$ is such that $h/\abs{h}=1$.
\end{proof}

%

\subsection{Examples}

Below we provide a list of solutions to \eqref{Eq3} obtained
by application of propositions \ref{ConstrGen} and \ref{Construction}.

\begin{enumerate}

\item Let $\lr{x,y,z}$ be the standard Cartesian coordinate system. 
We assume that $P\in C^{\infty}\lr{\ov{\Omega}}$ is a function of the spatial coordinate $x$, $P=P\lr{x}$.
Let $c\in\mathbb{R}$ be a positive real constant such that $c> P$ in $\Omega$. Then, the vector field
\begin{equation}
\bol{w}=\sqrt{2}\left[\nabla\lr{\int{\sqrt{c-P}}\,dx}+y\,\nabla z\right],
\end{equation}
is a solution to \eqref{Eq3}. Furthermore, $\nabla\cdot\bol{w}=\sqrt{2}\frac{\p}{\p x}\lr{\sqrt{c-P}}$, which vanishes for $P=c-k^2$, $k\in\mathbb{R}$. $y$ is an invariant of the flow generated by $\bol{w}$.

\item Let $\lr{x,y,z}$ be the standard Cartesian coordinate system.
We assume that $P\in C^{\infty}\lr{\ov{\Omega}}$ is a function of the spatial coordinates $x$ and $y$, $P=P\lr{x,y}$. 
A local solution of \eqref{Eq3} can be obtained by setting $C=\sqrt{2}\,z$, and by solving the following equations for $\mu=\mu\lr{x,y}$ and $\lambda=\lambda\lr{x,y}$ (see proposition \ref{Construction}),
\begin{subequations}
\begin{align}
&\frac{1}{2}\abs{\nabla\mu}^2=c-P\lr{x,y},\\
&\nabla\mu\cdot\nabla\lambda=0.
\end{align}
\end{subequations}
Then, the vector field
\begin{equation}
\bol{w}=\nabla\mu+\sqrt{2}\,\lambda\,\nabla z,
\end{equation}
is the desired solution. Furthermore, $\lambda$ is an invariant of the flow generated by such $\bol{w}$. For example, if $P=-\exp\lr{x+y}$, we obtain
\begin{equation}
\bol{w}=2\,\nabla\left[\exp\lr{\frac{x+y}{2}}\right]+\sqrt{2}\,\lr{x-y}\,\nabla z.
\end{equation}
We have $\di\bol{w}=\exp\lr{\frac{x+y}{2}}\neq 0$.

\item Let $\lr{r,\phi,z}$ be the cylindrical coordinate system introduced in section 3.
We assume that $P\in C^{\infty}\lr{\ov{\Omega}}$ is a function of the radial coordinate, $P=P\lr{r}$. Let $c\in\mathbb{R}$ be a positive real constant such that $c> P$ in $\Omega$. Then, the vector field
\begin{equation}
\bol{w}=\sqrt{2}\left[\nabla\lr{\int{\sqrt{c-P}}\,dr}+\phi\,\nabla z\right],
\end{equation}
is a solution to \eqref{Eq3}. Furthermore, $\nabla\cdot\bol{w}=\sqrt{2}r^{-1}\frac{\p}{\p r}\lr{r\sqrt{c-P}}$, which vanishes for $P=c-k^2/r^2$, $k\in\mathbb{R}$. $\phi$ is an invariant of the flow generated by $\bol{w}$.

\item Let $\lr{r,\phi,z}$ be the cylindrical coordinate system of the previous example.
We assume that $P\in C^{\infty}\lr{\ov{\Omega}}$ is a function of the coordinate $z$, $P=P\lr{z}$. Let $c\in\mathbb{R}$ be a positive real constant such that $c> P$ in $\Omega$. Then, the vector field
\begin{equation}
\bol{w}=\sqrt{2}\left[\nabla\lr{\int{\sqrt{c-P}}\,dz}+\phi\,\nabla r\right],
\end{equation}
is a solution to \eqref{Eq3}. $\phi$ is an invariant of the flow generated by $\bol{w}$.
We have $\di\bol{w}=-\frac{P_{z}}{\sqrt{2\lr{c-P}}}+\sqrt{2}\frac{\phi}{r}$.

\item Let $\lr{r,\phi,z}$ be the cylindrical coordinate system of the previous example. We assume that $P\in C^{\infty}\lr{\ov{\Omega}}$ is a function of the coordinates $r$ and $\phi$, $P=P\lr{r,\phi}$. 
A local solution of \eqref{Eq3} can be obtained by setting $C=\sqrt{2}\,z$, and by solving the following equations for $\mu=\mu\lr{r,\phi}$ and $\lambda=\lambda\lr{r,\phi}$ (see proposition \ref{Construction}),
\begin{subequations}
\begin{align}
&\frac{1}{2}\abs{\nabla\mu}^2=c-P\lr{r,\phi},\\
&\nabla\mu\cdot\nabla\lambda=0.
\end{align}
\end{subequations}
Then, the vector field
\begin{equation}
\bol{w}=\nabla\mu+\sqrt{2}\,\lambda\,\nabla z,
\end{equation}
is the desired solution. Furthermore, $\lambda$ is an invariant of the flow generated by such $\bol{w}$. For example, if $P=-\frac{1}{2}e^{-2r}\left[1+\cos^2{\phi}\lr{r^{-2}-1}\right]$, we obtain
\begin{equation}
\bol{w}=\nabla\left[e^{-r}\sin\phi\right]+\sqrt{2}\,e^{1/r}\cos\phi\,\nabla z.
\end{equation}
We have $\di\bol{w}=\frac{e^{-r}}{r}\lr{\sin\phi-r^{-2}\cos{\phi}}$.



\end{enumerate}

\section{Concluding remarks}

In this paper, the local theory of solution for steady ideal fluid and magneto-fluid equilibria
was studied. A local theory provides geometric conditions for the existence 
of solutions in a small neighborhood around a chosen point in the domain of interest.
Due to the appearance of solenoidal and non-integrable vector fields in the 
fluid equations, the local theory relies on the Lie-Darboux theorem of differential geometry, 
and the Frobenius integrability conditions for smooth differential forms of order 1. 
 
Here, we extended the local theory of representation and construction of Beltrami fields
developed in \cite{Sato} to solenoidal Beltrami fields, generalized Beltrami fields, ideal MHD equilibria, and general steady ideal Euler flows.
Regarding solenoidal Beltrami fields, we proved a theorem 
that enables
the construction of harmonic orthogonal coordinates. Such coordinates can be used to
obtain families of solenoidal Beltrami fields. Several analytic examples pertaining to Cartesian, 
cylindrical, and spherical geometry were given.
We discussed the existence of singular solutions to the boundary value problem
for solenoidal Beltrami fields, and found that, while regular solutions preserving the local representation
 are not admissible, if the metric coefficients of the relevant coordinate system satisfy certain geometric conditions, singular solutions exist. An example of singular solution in a spherical domain was given. 
We further showed that the problem of solving for a generalized Beltrami field
can always be reduced to the problem of finding a Beltrami field, 
and obtained
explicit examples of generalized Beltrami fields in Cartesian and cylindrical coordinates.
We proved a local representation theorem for ideal MHD equilibria, 
and
showed that, locally, they can always be represented as the solution of a pair of
coupled second order partial differential equations. 
Finally, we examined local conditions for existence of solutions in the case of
steady ideal Euler flows, and showed that they can be locally represented as a Hamiltonian system with two degrees of freedom. 
Several explicit solutions were given.


\section{Acknowledgments}

\noindent The research of N. S. was supported by JSPS KAKENHI Grant No. 18J01729, 
and that of M. Y. by JSPS KAKENHI Grant No. 17H02860.


\end{document}